\documentclass[journal]{IEEEtran}

\usepackage[dvips]{graphicx}
\usepackage{amsmath}
\usepackage{dsfont}
\usepackage{algorithmic}
\usepackage{fixltx2e}
\usepackage{stfloats}
\usepackage{graphicx}
\usepackage{romannum}
\usepackage{gensymb}
\usepackage{cite}
\usepackage{amsmath,amssymb,amsfonts}
\usepackage{algorithmic}
\usepackage{graphicx}  
\usepackage{subcaption} 
\usepackage{textcomp}
\usepackage{xcolor}
\usepackage{lettrine}
\usepackage{mathtools}  
\usepackage[english]{babel}
\usepackage{amsmath,amsthm}
\usepackage{multirow}
\usepackage{caption}

\newtheorem{corollary}{Corollary}
\newtheorem{lemma}{Lemma}
\theoremstyle{definition}
\newtheorem{remark}{Remark}
\newtheorem{definition}{Definition}
\newtheorem{note}{Note}
\newtheorem*{notation}{Notation}

\newcommand{\nick}[1]{{\color{black} {#1}}}
\newcommand{\MK}[1]{{\color{black} {#1}}}
\newcommand{\MKK}[1]{{\color{black} {#1}}}
\newcommand{\Mk}[1]{{\color{black} {#1}}}

\ifCLASSINFOpdf
 
\else

\fi

\begin{document}


\title{Coexistence of Satellite-borne Passive Radiometry and Terrestrial NextG Wireless Networks in the Restricted L-Band}


\author{Mohammad Koosha,~\IEEEmembership{Student Member,~IEEE,}
        Nicholas Mastronarde,~\IEEEmembership{Senior Member,~IEEE}
        \thanks{The work of M. Koosha and N. Mastronarde was supported in part by the NSF under Award \#2030157.}
}

\markboth{this article is under consideration for publication in an IEEE Transactions}%
{Shell \MakeLowercase{\textit{et al.}}: Bare Demo of IEEEtran.cls for IEEE Journals}

\maketitle

\begin{abstract}

 The rapid growth of active wireless communications technologies has fostered research on spectrum coexistence worldwide.
 One idea that is gaining attention is sharing frequency bands solely devoted to passive  applications, such as passive remote sensing. One such option is the 27 MHz L-band spectrum from 1.400 to 1.427 GHz. 
Active wireless transmissions are prohibited in this passive band due to radio regulations aimed at preventing Radio Frequency Interference (RFI) on highly sensitive passive radiometry instruments. The Soil Moisture Active Passive (SMAP) satellite, launched by the National Aeronautics and Space Administration (NASA), is a recent space-based remote sensing mission that passively scans the Earth's electromagnetic emissions in this 27 MHz band to assess soil moisture on a global scale periodically. \nick{In this paper, leveraging stochastic geometry, we assess the feasibility of using this passive band for the downlink (base station to user equipment) of a large-scale terrestrial NextG wireless network exposed to SMAP, while ensuring that the error induced on SMAP's measurements due to RFI is below a given threshold.}
Although the methodology established here is based on SMAP's specifications, it is adaptable to various passive sensing satellites regardless of their orbits or operating frequencies. 

\end{abstract}


\IEEEpeerreviewmaketitle

\section{Introduction}
\IEEEPARstart{B}{esides} those at absolute zero Kelvin, all objects emit electromagnetic radiation due to their body temperature (natural thermal emissions) \cite{campbell2011introduction}. 
Remote Sensing (RS) instruments measure thermal radiation from land, oceans, atmosphere, and vegetation, which provides invaluable information about soil moisture (SM), sea surface salinity (SSS), and other climatological and geological variables.
As climate change escalates, this information becomes even more critical in many applications, including weather, flood, and drought forecasting, as well as agricultural productivity, carbon cycle, and hydrological modeling \cite{entekhabi2010soil}.

Although thermal emissions from objects exist on a wide range of frequencies, some frequency bands are more favorable for different RS applications. For example, L-band spectrum (1--2 GHz) is optimal for sensing SM because of its relatively high sensitivity to the moisture content in deeper soil, its semi-transparency to vegetation covering the soil, and lower attenuation from the atmosphere\cite{o2018algorithm}. Such favorable properties have inspired the allocation of several portions of the frequency spectrum for various passive RS applications, including the L-band (1.4 GHz), C-band (6 GHz), X-band (10 GHz), and mm-wave frequencies ($>$30 GHz)~\cite{national2007handbook}. For some passive bands, Federal Communications Commission (FCC) radio regulations disallow in-band active wireless transmissions and out-of-band spectral leakage from adjacent bands \cite{national2007handbook}.

In parallel, the proliferation of NextG networks demands enhanced, reliable, and ubiquitous connectivity to support various applications with faster speeds, reduced latency, and massive numbers of connected devices. This will require improving utilization of the frequency spectrum \cite{hurtado2022deep} and will further exacerbate the so-called spectrum crunch~\cite{sicker2015wireless}. 
With the limited availability of radio frequency spectrum, network operators and regulatory bodies face the task of efficiently allocating and managing this valuable resource.

One solution gaining attention is the use of passive bands for active wireless communications \cite{polese2023coexistence}.  However, this increases the risk of active wireless devices inducing radio frequency interference (RFI) on passive sensing instruments. 
Since these instruments are designed to measure faint natural emissions, even weak RFI can lead to erroneous measurements. Worse still, strong RFI can potentially saturate the passive radiometry electrical equipment, rendering it inoperable. This conflict between active and passive users of the spectrum calls for research on active-passive coexistence in spectrally adjacent bands or in co-channel (i.e., in the same frequency band). An indispensable step in this course is evaluating the amount of RFI that active NextG devices induce on passive sensing instruments, which can then help us determine the degree to which NextG devices can utilize passive bands without compromising the incumbent passive sensing instruments' operations and the science that relies on their measurements.

Among various methodologies, space-borne passive sensing has gained significant importance because it can provide comprehensive coverage in a short time \cite{entekhabi2010soil}. A notable benefit is the ability to deploy a single sensor globally and for an extended duration. This allows for more accurate monitoring of changes in the observed natural phenomena across different geographical areas and time spans. It also overcomes the limitations of using multiple instruments with different calibrations, ensuring greater precision in data tracking.

\nick{In this study, we develop a mathematical framework to model the RFI originating from the down-link of a large-scale NextG terrestrial cellular network and its impact on an Earth Exploration Satellite Service (EESS) satellite.\footnote{\Mk{In this paper, we assume that the network operates in Frequency Division Duplex (FDD) mode; therefore, User Equipment (UE) devices utilize a separate RF band for the uplink channel and do not cause RFI at SMAP. However, our proposed methodology can be easily adapted to assess the RFI on EESS satellites caused by UE emissions on the uplink channel.}}  
For illustration, we develop our model based on the National Aeronautics and Space Administration (NASA) Soil Moisture Active Passive (SMAP) satellite \cite{entekhabi2014smap}, which is one of the latest RS satellites active in the restricted L-band ($1400-1427$ MHz). According to FCC radio regulations \cite{national2007handbook}, this band is solely devoted to passive radiometry for \textit{Earth Exploration Satellites}, \textit{Radio Astronomy}, and \textit{Space Research}, thus in-band active wireless transmissions and out-of-band electromagnetic emissions are strictly prohibited in this band. In this context, our contributions are as follows:}

\begin{itemize}
	\item \nick{We use stochastic geometry to model the aggregate RFI originating from the downlink of a large-scale terrestrial NextG network that we imagine operates co-channel with SMAP. We employ a Thomas cluster process to model the distribution of cellular base stations on Earth; shadowed Rician fading to model the Earth-to-space channel; and a sectorized antenna model to capture the beamwidth and main/side-lobe gains of SMAP's antenna.}
	
	\item \nick{We derive the moment generating function (MGF) of the aggregate RFI incident on both the main- and side-lobes of SMAP's antenna. We then derive the average, variance, and higher central moments of the RFI from the corresponding cumulants. This model effectively showcases how the cluster density, the number of active base stations per cluster, SMAP's antenna gains, and the channel parameters influence the aggregate RFI.}
	
	\item \nick{We demonstrate that, while it is crucial to avoid RFI on SMAP's main-lobe due to its high intensity and variance, the extremely low side-lobe gain of SMAP's antenna dramatically reduces the RFI's intensity and (more importantly) its variance, resulting in RFI that is quite predictable. Based on this insight, we propose a simple RFI mitigation technique based on subtracting the expected RFI on SMAP's side-lobe from SMAP's RFI-contaminated measurements. We then derive an upper bound on the probability that the errors in the corrected measurements exceed an RFI threshold. We refer to this probability as the \textit{Sensing Outage Probability (SOP)}.}
	
	\item \nick{We use Monte Carlo simulations to validate our analytical models of the average intensity and variance of the RFI incident on the main- and side-lobes of SMAP's antenna. Additionally, we illustrate that, while it is imperative to avoid RFI on SMAP's main-lobe, a substantial number of clusters can coexist while exposed to the side-lobe without exceeding the tolerable RFI threshold  specified in SMAP's documentation. Furthermore, we compare the analytical upper bound of the SOP with the estimated SOP from Monte Carlo simulations. This reveals that the upper bound is loose and therefore the analytical model \textit{underestimates} the number of BSs that can safely operate/coexist co-channel with SMAP.}
	
	\item \nick{Lastly, we assess the spectral efficiency and achievable sum throughput (across all clusters and BSs within the network) and show the impact of the resulting RFI on SMAP's SOP. Our findings demonstrate a substantial acquired throughput while ensuring integrity of SMAP's measurements, underscoring the robustness of our approach in managing RFI. 
    }
	
\end{itemize}

The remainder of this paper is organized as follows. In Section~\ref{sec:related_work}, we present related work. In Section~\ref{sec:methodology}, we review our methodology for aggregate RFI analysis and introduce \nick{the system model}.
\nick{In Section~\ref{sec:RFI_analysis}, we present our RFI analysis. In Section~\ref{sec:spectral_efficiency}, we analyze the spectral efficiency and sum throughput of the terrestrial network.}
\nick{We present our simulation results in Section~\ref{sec:results} and conclude the paper in Section~\ref{sec:conclusion}.}

\begin{notation}
    We denote random variables by upper-case letters and their realizations or other deterministic quantities by lower-case letters. We use bold font to denote vectors and normal font to denote scalar quantities. For instance, $X$ and $\bold{X}$ denote a one-dimensional (scalar) random variable and a vector random variable, respectively. Similarly, $x$ and $\bold{x}$ denote a \MK{deterministic} scalar and vector, respectively.
\end{notation}
\section{Related Work}\label{sec:related_work}
In this section, we provide a concise overview of the most relevant research concerning aggregate RFI analysis in remote sensing applications, \nick{coexistence among active and passive spectrum users}, as well as the analysis of RFI on satellites.

Polese et al. \cite{polese2023coexistence} investigate the allocation \nick{of spectrum}
above 100 GHz, where multiple narrow passive sensing sub-bands scatter over the frequency spectrum, precluding the \nick{exclusive} allocation of large continuous bandwidth chunks \nick{to active wireless communications systems}. 
In their work, they consider a wireless link between two terrestrial base stations that interfere with an EESS satellite. Using International Telecommunication Union (ITU) recommendations, they develop a path loss model between the base stations and the EESS satellite. 
In contrast, in this paper, we consider the coexistence of a \textit{large-scale} network \nick{exposed to an} EESS satellite \nick{and the resulting \textit{aggregate} RFI at the satellite.}

Zheleva et al. \cite{zheleva2023radio} propose the idea of Radio Dynamic Zones (RDZ) as experimental testbeds for spectrum research. These testbeds enable the study of the coexistence of different technologies used by three major frequency spectrum stakeholders: consumer broadband, microwave remote sensing, and radio astronomy. The RDZs are regional-scale geographic areas of $10$s to $100$s of square kilometers, where these diverse stakeholders can conduct research on spectrum coexistence.

In \cite{ramadan2017new}, authors investigate spectrum coexistence between \Mk{terrestrial} Cellular Wireless Communications (CWC) and terrestrial Radio Astronomy Systems (RAS). They propose a geographical Shared Spectrum Access Zone (SSAZ) around a RAS site in which there is a three-phase spectrum access procedure for the surrounding cellular network based on the distance to the RAS site. Also, the network outside the SSAZ has full spectrum access. They consider both hexagonal CWC cells and a 2-dimensional Poisson Point Process to model the cellular network surrounding the RAS site. 

\Mk{In \cite{testolina2024modeling}, Testolina et al. investigated the impact of 6G terrestrial RFI on passive sensing EESS satellites, combining an analytical model with large-scale simulations. Their study highlighted the interactions among directional communications, ground reflections, satellite orbits, sensor orientations, and sub-THz propagation, revealing that narrow beams alone cannot prevent RFI. The simulations, involving up to $10^5$ nodes over $85$ km² and real sensor deployments from EESS missions, demonstrated significant interference potential from ground reflections and dense sub-THz networks. The research also showed that atmospheric attenuation and building obstructions can effectively shield passive users more than directional arrays.} In \cite{park2019modeling}, authors investigate the impact of the large tail in OFDM based systems to evaluate the out-of-band RFI caused in xG systems.

\MKK{In \cite{koosha2022opportunistic}, we investigated opportunistic temporal spectrum coexistence between a large-scale NextG terrestrial network and SMAP during periods when a NextG site is \textit{not exposed} to SMAP, i.e., when SMAP is not within the sky above the NextG site. In contrast, this paper utilizes stochastic geometry to model the impact of aggregate RFI induced on SMAP when it \textit{is exposed} to a large-scale NextG terrestrial network.}

\section{Methodology}\label{sec:methodology}
\subsection{SMAP \& Brightness Temperature}

As depicted in Figure~\ref{fig:smap_characteristics}, SMAP has a 6-meter-wide conically-scanning golden mesh reflector with a 3-dB antenna beam-width of $2.4{\degree}$ that projects a footprint of roughly $40 \times 40$ km$^2$ with an Earth incident angle of $40{\degree}$ from an altitude of $685$ km. An Ortho-Mode Transducer (OMT) feedhorn collects radiation from the mesh reflector and separates it into \textit{vertical} and \textit{horizontal} polarizations. Figure~\ref{fig:SAG} shows a 2-dimensional cross-section of SMAP's antenna gain for the vertical polarization.
Through sectorization, which is a common method in stochastic geometry, we represent SMAP's antenna gain for each polarization $(p)$ as:
\begin{equation}
g=\left\{
\begin{array}{ll}
    g_{(ml)} , & \mbox{if } |\Mk{\delta}| \leq 1.2{\degree},  \\
    g_{(sl)}, & \mbox{if } |\Mk{\delta}| > 1.2{\degree},
\end{array}    
\right. \label{eq:G_SMAP}
\end{equation}
where $g_{(ml)}$ and $g_{(sl)}$ denote the \textit{main-lobe} and \textit{side-lobe} gains, respectively, and $|\Mk{\delta}|$ is the deviation from the main-lobe axis. For each polarization $(p)$, SMAP separately captures the \textit{brightness temperature} of soil, $t_{soil}^{(p)}$ (in Kelvin), from the antenna footprint by capturing the soil's natural passive thermal radiation. These brightness temperature measurements can be translated to soil moisture content using models like the Tau-Omega model \cite{de2015converting}. We use the Nyquist noise formula \cite{turner2012johnson} to convert electromagnetic power to brightness temperature as:
\begin{equation}
t^{(p)}=\frac{\mathrm{p}^{(p)}}{k_{b}\beta},
\end{equation}
where $\mathrm{p}^{(p)}$ is the electromagnetic power received by polarization $(p)$, $k_b$ is the \textit{Boltzmann} constant, and $\beta$ is the radio frequency \textit{bandwidth}. 

\begin{note}
    Due to the symmetrical nature of SMAP's antenna gain for both polarizations, as well as the symmetry in the RFI scenario, we assume identical RFI characteristics for both of SMAP's polarizations. Thus, the discussions that follow hold true for SMAP's measurements in both polarizations.
\end{note}

\begin{figure}
\centering
  \includegraphics[width=\linewidth]{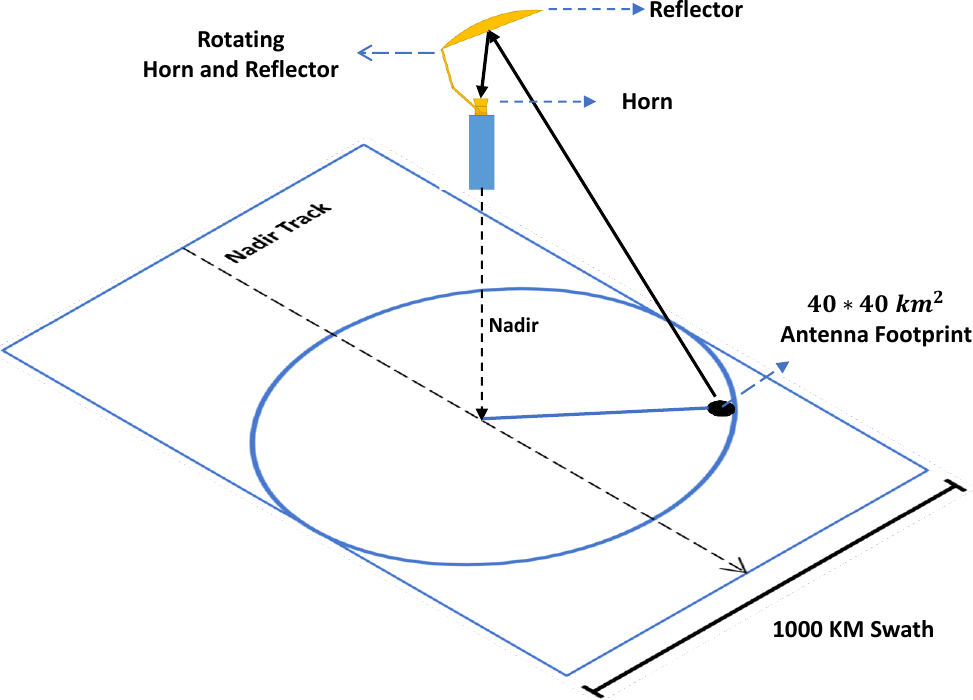}
  \caption{Horn and reflector rotation, beam footprint of the reflector, and SMAP's nadir track.}
  \label{fig:smap_characteristics}
\end{figure} 

\begin{figure}
\centering
  \includegraphics[width=\linewidth]{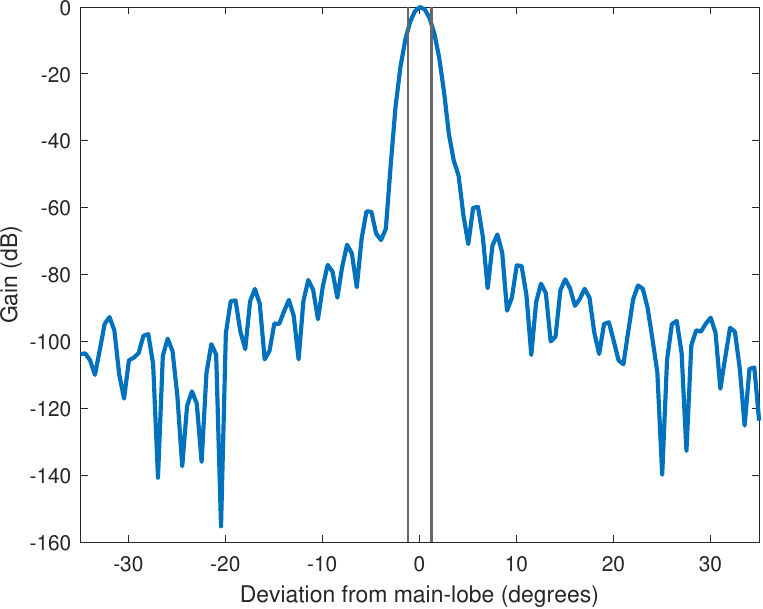}
  \caption{A 2D cut of SMAP's conical antenna gain for the vertical polarization. The gain for the horizontal polarization is similar. The two vertical lines show the $2.4^{\degree}$ beam-width.}
  \label{fig:SAG}
\end{figure}

\subsection{Methodology for RFI Analysis}
In this section, we provide a concise explanation of the underlying logic guiding our analysis of RFI on SMAP. SMAP's measurements for each polarization can be seen as:
\begin{equation}
    \Mk{T}_{meas}=t_{soil}+T_{RFI} \label{eq:T_meas},
\end{equation}
where $T_{RFI}$ denotes the RFI temperature at SMAP. Since $T_{RFI}$ is a random variable, it causes uncertainty in SMAP's measurements. According to SMAP's documentation, uncertainties below threshold value $\tau = 1.3$ K are acceptable for SMAP's measurements \cite{rajabi2019characteristics}. To model $T_{RFI}$ from the downlink of a large terrestrial network, we imagine a set of $\{i\}$ Base Station (BS) clusters on the Earth-cut exposed to SMAP, where each cluster $i$ comprises a set of $\{j\}_i$ BSs, each BS$_{ij}$ has a (maximum) total electromagnetic transmission power $p_{tx}$, and $P_{ij}$ is the amount of power received by SMAP from BS$_{ij}$. Clusters $\{i\}_{(ml)} \subset \{i\}$ are located on SMAP's main-lobe antenna footprint and clusters $\{i\}_{(sl)} \subset \{i\}$ are exposed to SMAP's side-lobe. Accordingly, $T_{RFI}$ in \eqref{eq:T_meas}, can be decomposed into its main- and side-lobe components as:
\begin{align}
    T_{RFI} &= T_{(ml)}+T_{(sl)}, \label{eq:T_RFI}
\end{align}
where 
\begin{equation}
    T_{(l)} = \sum\nolimits_{\{i\}_{(l)}}\sum\nolimits_{\{j\}_i} T_{ij}, \quad \text{with} \quad T_{ij} = \frac{P_{ij}}{k_{b}\beta}, \label{eq:T_RFI_l} 
\end{equation}
where $(l)$ is either $(ml)$ or $(sl)$. 
\Mk{Using the following path loss model we have:}
\begin{equation}
    P_{ij} = \frac{1}{2}g \left( \frac{c}{4\pi f} \right)^2 \MKK{x}_{ij}^{-\alpha} \lvert H_{ij} \lvert^{2} p_{tx} \label{eq:P_ij}
\end{equation}
and according to \eqref{eq:T_RFI_l}:
\begin{equation}
    T_{ij} = g \omega \lvert H_{ij} \lvert^{2} \MKK{x}_{ij}^{-\alpha}, \label{eq:T_ij} 
\end{equation}
where $\omega= \frac{p_{tx}}{2k_{b}\beta} \left( \frac{c}{4\pi f} \right)^2$; $g$ is the gain of SMAP's antenna (based on \eqref{eq:G_SMAP}, $g = g_{(ml)}$ if $i \in \{i\}_{(ml)}$ and $g = g_{(sl)}$ if $i \in \{i\}_{(sl)}$); $c$ is the  speed of light; $f$ is the frequency; $\MKK{x}_{ij}$ is the distance of BS$_{ij}$ to SMAP; $\alpha>2$ is the path loss exponent, $H_{ij}$ is the channel gain between BS$_{ij}$ and SMAP; and the coefficient $1/2$ in \eqref{eq:P_ij} indicates that half the electromagnetic power is absorbed by each polarization. \Mk{Here, we neglect atmospheric path loss since it is known to be negligible in the L-band (in fact, this is one reason why the L-band is used for satellite-borne passive radiometry \cite{entekhabi2014smap}). 
}

\subsection{RFI Mitigation, Cumulants, \& Sensing Outage Probability} \label{sec:mitigation}
The utilization of raw and central moments in signal analysis is widely employed for statistical detection and mitigation (D/M) of RFI in passive radiometry \cite{querol2019review, bringer2021properties}. For instance, natural electromagnetic emissions from Earth typically are below $330$ K, whereas most man-made RFI emissions exceed $500$ K \cite{querol2019review}. Moreover, moments play a crucial role in assessing the normality of a signal. While natural thermal emissions are anticipated to follow zero-mean Gaussian distributions, man-made RFI manifests non-Gaussian characteristics, such as non-zero skewness and non-$3$ kurtosis. \textit{Skewness}, derived from the third central moment, is a measure of the asymmetry of a random variable's probability distribution around its mean. On the other hand, \textit{kurtosis}, derived from the fourth central moment, is a measure of the distribution’s tailedness.
To acquire the moments of RFI $T_{(ml)}$ and $T_{(sl)}$ in \eqref{eq:T_RFI}, we use the so-called \textit{Cumulant Generating Function} (CGF).

\begin{definition}\label{def:CGF}
The CGF of random variable $X$ is defined as:
\begin{equation}
    K(\eta) = \log \mathds{E}_X[e^{\eta X}],
\end{equation}
which is the $\log$ of the \textit{Moment Generating Function} (MGF) $M(\eta) = \mathds{E}_X[e^{\eta X}]$ of random variable $X$. Given its CGF, the $n$th \textit{cumulant} of $X$ can be obtained as follows:
\begin{equation}
    k_{n} = K^n(0),
\end{equation}
where $K^n(\eta)$ denotes the $n$th derivative of $K(\eta)$.
\end{definition} 

  \vspace{-3mm}
\begin{remark}
    For random variable $X$, $k_1$ equals the first \textit{raw} moment of $X$, i.e., $k_1=~\mathds{E}[X]$. For $n \in \{ 2,3\}$, $k_n=\mu_n$, where $\mu_n =\mathds{E}[(X-\mathds{E}[X])^n]$ is the $n$th \textit{central} moment of $X$. Consequently, the variance of $X$ is its $2$nd central moment, i.e., $k_2 = \mu_2$. Lastly, higher order central moments $\mu_{n}$ for $n>3$ can be acquired by a combination of cumulants of $X$. For example, $\mu_4=k_4+3(k_2)^2$.
\end{remark}

\MKK{As we will show through our results in Section \ref{sec:results}, the RFI brightness temperature imposed on SMAP's main-lobe exhibits very high intensity and variability, often reaching hundreds of Kelvins of standard deviation. \nick{Therefore, it is essential that protocols are in place to prevent RFI on SMAP's main-lobe.}\footnote{\nick{This could potentially be achieved by leveraging knowledge of SMAP's orbit, instantaneous location, and antenna footprint locations (which can all be tracked \cite{koosha2022opportunistic}) to silence all BSs with LoS to SMAP's main-lobe; however, this is beyond the scope of this paper.}} In contrast, due to SMAP's extremely low side-lobe gain, the RFI impact on the side-lobe is both weaker and more predictable, with variations typically below a few Kelvin.}

\nick{While avoiding RFI on SMAP's main-lobe, we propose to use}
the first raw moment of RFI brightness temperature on SMAP's side-lobe to correct SMAP's RFI-contaminated measurements as follows:
\begin{equation}
    \hat{\Mk{T}}_{soil}=\Mk{T}_{meas}-\mathds{E}[T_{(sl)}], \label{eq:mitigation}
\end{equation}
where $\Mk{T}_{meas}$ is defined in \eqref{eq:T_meas}. This estimate yields an error of $Err = \hat{\Mk{T}}_{soil} - t_{soil} = T_{\MKK{(sl)}}-\mathds{E}[T_{(sl)}]$ in  SMAP's measurements. According to SMAP's documentation, absolute error values less than 1.3 K are acceptable for SMAP's measurements. To generalize this idea, we focus our analysis on the \textit{Sensing Outage Probability} (SOP) for an arbitrary error threshold $\tau > 0$, which we define as:
\begin{equation}
    SOP(\tau)\triangleq \mathds{P}(|Err|>\tau). \label{eq:SOP}
\end{equation}
 We use Chebyshev's inequality to derive an upper-bound on $SOP(\tau)$ as follows: 
\begin{equation}
    SOP(\tau) = \mathds{P}(|T_{\MKK{(sl)}}-\mathds{E}[T_{(sl)}]|>\tau)< \frac{\mu_n}{\tau^{n}}, \quad n=2i, i \in \mathbb{N} \label{eq:SOP_up_bound}
\end{equation}
where $\mu_{n}$ is the $n$th central moment of $T_{(sl)}$. In the following sections, we show that the SOP can be kept low due to the extremely low side-lobe gains of SMAP's antenna. Note that, to acquire the cumulants of RFI on SMAP's main- and side-lobes from Definition \ref{def:CGF}, we first have to acquire their MGFs.
\Mk{
\begin{note}
    We acknowledge that Chebyshev's inequality does not provide a tight bound. However, in practical applications, obtaining the exact value of $\mathds{E}[T_{(sl)}]$ can be challenging and may lack precision. Consequently, we opt for a more conservative (loose) bound to accommodate potential inaccuracies in the estimation of $\mathds{E}[T_{(sl)}]$.
\end{note}
}

\subsection{Geometric Assumptions}
As depicted in Figure~\ref{fig:smap_exposed}, the Earth's center is the \textit{origin} $(0,0,0)$ and SMAP is located at the point $\bold{h}=(0,0,h)$, where $h$ is the distance of SMAP from the Earth's center. 
\begin{figure}
\centering
  \includegraphics[width=\linewidth]{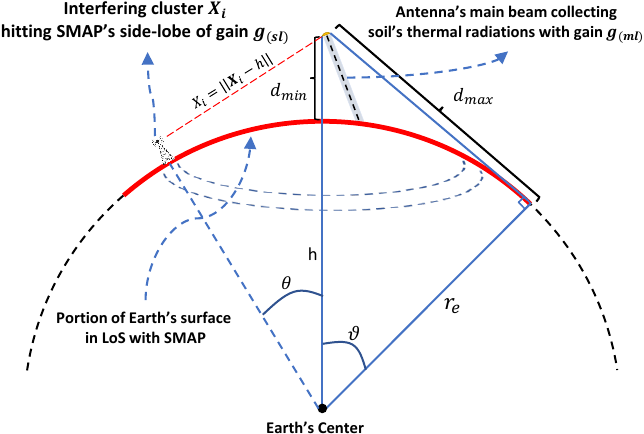}
  \caption{2D representation of the Earth-cut (red curve) exposed to SMAP. 
  }
  \label{fig:smap_exposed}
\end{figure}The area exposed to the satellite, shown with the red cap 
 and encircled in $\theta \in [0, \vartheta]$ in Figure~\ref{fig:smap_exposed}, is defined with the \textit{Borel-set} $\mathcal{B}=\left \{\lVert   \bold{x}  \lVert = r_e, \; \frac{  \bold{x}  \cdot \bold{h}}{\lVert \bold{x} \lVert \lVert \bold{h} \lVert} \geq \cos{(\vartheta)}  \right\}$ in $\mathbb{R}^{3}$ measure space, where $r_e$ is the Earth's radius, $\lVert \cdot \lVert$ is the Euclidean norm, and $\cos{(\vartheta)}=\frac{r_e}{h}$. We define $\mathcal{B}_{(ml)} \subset \mathcal{B}$ as SMAP's antenna footprint projected on Earth 
 and $\mathcal{B}_{(sl)}=\mathcal{B} \big\backslash \mathcal{B}_{(ml)}$ as the area exposed to the side-lobe of SMAP's antenna. 
 
 Let $\Psi = \{\bold{X}_i, i \in \mathbb{N}\} \subset \mathcal{B}$ denote a parent \textit{Poisson Point Process} (PPP) with intensity measure $\Lambda(d\bold{x})$, where $\bold{X}_i$ denotes a cluster center. For simplicity, we use a \textit{homogeneous} PPP such that $\Lambda(d\bold{x})=\lambda_c$. According to Figure \ref{fig:smap_exposed}, the minimum distance of a cluster center to SMAP is $d_{min}=h-r_e$, while the maximum distance is $d_{max}=\sqrt{h^2-r_e^2}$. With $\psi=\{\bold{x}_i\}$ defined as a realization of $\Psi$, for each point $\bold{x} \in \psi$, we associate an \textit{independent and identically distributed} (i.i.d) offspring PPP $\Xi_\bold{x}$. Each cluster $\Xi_\bold{x}$ consists of $N$ i.i.d random points (BSs) \MKK{\{$\bold{Y}\} \Mk{ \subset \mathcal{B} }$ with \textit{probability density function} (PDF) $g_\bold{Y}(\bold{y})$ \Mk{from the cluster center}, where $N \sim \textsf{Pois}{(\lambda_{BS})}$. }
 We define $\Phi \sim  \mathcal{P} (\lambda_c, g_\bold{Y},\lambda_{BS})$, where
\begin{equation}
    \Phi=\bigcup\nolimits_{\bold{x} \in \Psi} \left(\bold{x} + \Xi_{\bold{x}} \right).
\end{equation}
Accordingly, the cluster centers located in SMAP's main- and side-lobes are respectively $\Psi_{(ml)}=\Psi(\mathcal{B}_{(ml)})$ and $\Psi_{(sl)} = \Psi \big \backslash \Psi_{(ml)}$, and their respective clusters are $\Phi_{(ml)}$ and $\Phi_{(sl)}$.

\begin{note} \label{note:BS_locations}
    Let $X_i= \lVert \bold{X}_i - \bold{h} \lVert$ denote the distance of terrestrial cluster center $\bold{X}_i$ to SMAP and let $x_i$ denote its realization. We assume that the cluster's dispersion is much less than $x_i$ ($\ll x_i$), since an urban area is on the order of a few kilometers, while the cluster's distance to SMAP is at least $d_{min}=685$ km. 
    Thus, for simplicity of analysis, we assume that the off-spring (BSs) within each cluster \nick{are all located at the cluster's center and are therefore equidistant (with distance $x_i$) to SMAP}. \nick{In Section~\ref{sec:avg_std_RFI_results}, using Monte Carlo simulations, we show that this assumption has a negligible effect on the average and standard deviation of the RFI brightness temperature.}
\end{note}

\subsection{Channel Model} \label{sec:Channel_Model}


\nick{In this section, we define the channel model between a BS and SMAP (denoted by $H$ in \eqref{eq:P_ij}). This model accounts for both large-scale shadowing from terrain and buildings, and fading from multi-path effects. Before we can define the channel model, we need to identify the different components of a BS's emissions
that impact SMAP.} 
%
%
\nick{In the considered scenario, three significant elements are apparent. First, there is a potential LoS component due to the higher vertical side-lobes of a BS's antenna \Mk{(proof is omitted due to space limitations)}.
\noindent Second, in dense urban areas, this potential LoS component may be obstructed by terrain\MK{, e.g., tall buildings}. 
Third, there is multi-path RFI originating primarily from the main-lobe of a BS's antenna, which reflects off the Earth's surface because the antenna is angled downwards.} 

To capture these channel characteristics, we adopt the \textit{shadowed Rician fading} (SRF) model \cite{abdi2003new}, which is a proven channel model in land-satellite communication systems. In this model, the signal from a single source is \MK{characterized by three parameters $(\Omega, \;2b_0, \;m )$}: 
\nick{i) The parameter $\Omega$ is the average power of the primary LoS component, which follows a \textit{Nakagami-$m$} distribution; ii) the parameter $m$ can be fine-tuned to accommodate the probability of LoS signal obstruction; and iii) the parameter $2b_0$ is the average power of a scatter component, which follows a \textit{Rayleigh} distribution.} 
\Mk{As specified in SMAP's documentation, SMAP's sub-band integration period is $1.2~ms$ and full-band integration period is $300~\micro s$. 
In the context of these short integration periods, we use SRF because it is proven to account for small-scale fading as well as large-scale fading.
}
\section{RFI Analysis}\label{sec:RFI_analysis}
\MK{In this section, we analyze the RFI brightness temperature on SMAP's main- and side-lobes. We first derive the MGF of RFI brightness temperature at SMAP caused by one BS.} 
\MK{\subsection{MGF of RFI from One Base Station}}
Using \eqref{eq:T_ij}, the RFI brightness temperature for a BS located at cluster center $\bold{x}$ with  distance $x=\lVert \bold{x} -\bold{h} \lVert$ to SMAP is:

\begin{equation}
    T_{BS}(\bold{x}) = g \omega \lvert H \lvert^2 x^{-\alpha},   \label{eq:T_BS}
\end{equation}
where $H$ is the SRF channel model defined in Section \ref{sec:Channel_Model}.

\begin{corollary}\Mk{(MGF of $T_{BS}(\bold{x})$ under SRF channel).}
    Under shadowed Rician fading for the channel model $H$ with parameters $(\Omega, \;2b_0, \;m )$, the MGF of RFI brightness temperature at SMAP from one BS, defined in \eqref{eq:T_BS}, is:
    \begin{equation}
    M_{BS}(\eta;\; x, \;g) = \frac{(2b_{0}m)^{m}(1-g\omega x^{-\alpha}(2b_{0})\eta)^{m-1}}
    {[(2b_{0}m+\Omega)(1-g\omega x^{-\alpha}(2b_{0})\eta)- \Omega]^{m}}. \label{eq:MGF_BS_SRF}
\end{equation}
\end{corollary}
\begin{proof}
By setting $\eta := g \omega x^{-\alpha} \eta$ for the MGF of $\lvert H\lvert^2$ defined in \cite[Eq. (7)]{abdi2003new}, we acquire \eqref{eq:MGF_BS_SRF}.
\end{proof}
As we will discuss later, the series expansion of $M_{BS}(\eta;\; x, \;g)$ in \eqref{eq:MGF_BS_SRF} will help us acquire the cumulants of aggregate RFI brightness temperature on SMAP. 
As a first step, we first introduce the following Lemma.

\begin{lemma} \label{lemma:SE_MGF_H}\Mk{(Moments of SRF gain squared $\lvert H \lvert^2$).}
    The series expansion of the MGF of channel gain $\lvert H \lvert^2$ under shadowed Rician fading with parameters $(\Omega, \;2b_0, \;m )$, is:
    \begin{equation}
        M_{\lvert H \lvert^2} (\eta) = \sum\nolimits_{n=0}^{\infty} a_n \frac{\eta^n}{n!}, \label{eq:SE_MGF_H}
    \end{equation}
    where $a_n$ is the $n$th moment of $\lvert H\lvert^2$: i.e.,
    \begin{equation}
         \left( \frac{2b_{0}m}{2b_{0}m + \Omega} \right)^m (2b_0)^n \Gamma(n+1) F\left(m;n+1;1;\frac{\Omega}{2b_0m+\Omega}\right), \label{eq:A_n}
    \end{equation}
    where $\Gamma(\cdot)$ is the gamma function and $F(\cdot)$ is the hyper geometric function.
\end{lemma}
\begin{proof}
Refer to Appendix \ref{Appendix:SE_MGF_H}.    
\end{proof}
\begin{corollary}\Mk{(Series of MGF of $T_{BS}(\bold{x})$ under SRF).}
    The series expansion of $M_{BS}(\eta;\; x, \;g)$ in \eqref{eq:MGF_BS_SRF} is as follows:
    \begin{equation}
        M_{BS}(\eta;\; x, \;g) = \sum_{n=0}^{\infty} a_n (g\omega x^{-\alpha})^n \frac{\eta^n}{n!}, \label{eq:SE_MGF_BS_SRF}
    \end{equation}
    with $a_n$ defined in \eqref{eq:A_n}.
\end{corollary}
\begin{proof}
By setting $\eta := g \omega x^{-\alpha} \eta$ in \eqref{eq:SE_MGF_H}, we acquire \eqref{eq:SE_MGF_BS_SRF}.
\end{proof}

However, if we neglect the LoS component, e.g., due to very low LoS transmission power (small $\Omega$) or signal obstructions (small $m$), then $H$ reduces to a Rayleigh fading channel.

\begin{corollary} \Mk{(MGF of $T_{BS}(\bold{x})$ under Rayleigh channel).}
    Under Rayleigh fading, the MGF of RFI brightness temperature from one BS in \eqref{eq:MGF_BS_SRF} reduces to:
    \begin{equation}
   M_{BS}(\eta;\; x, \;g)=\frac{1}{1-g \omega x^{-\alpha}(2b_0)\eta}. \label{eq:BS_MGF_Rayleigh}
\end{equation}
\end{corollary}
\begin{proof}
    By setting $\Omega = 0$ \Mk{or $m=0$} in \eqref{eq:MGF_BS_SRF}, we acquire \eqref{eq:BS_MGF_Rayleigh}.
\end{proof}

\Mk{In a BS, the LoS component of RFI on a satellite mainly comes from the higher vertical side-lobe of the BS antenna. This interference may be negligible compared to the scatter component from the main-lobe of the BS antenna, which is tilted toward the ground. Therefore, the BS-satellite channel can be approximated using a pure Rayleigh model. Conversely, for RFI from the uplink channel of UE devices (UE-to-BS channel), the main lobe of a UE antenna can interfere with an EESS satellite, making an SRF channel model more suitable. Although this study assumes the network operates in FDD mode, with downlink channels (BSs) using the 27 MHz restricted L-band and UEs in out-of-band uplink channels, future work considering UEs operating co-channel with an EESS satellite can leverage our results.}

\subsection{MGF of RFI from One Cluster}
One key quantity that can help us determine the aggregate RFI brightness temperature at SMAP is the RFI brightness temperature contributed by \textit{one} cluster. Assuming the cluster is located at point $\bold{x}\in \psi$ and comprises $N$ BSs that are equidistant to SMAP (see Note \ref{note:BS_locations}), we have:
\begin{equation}
    T_{cluster}(\bold{x}) =\sum\nolimits_{j=1}^N T_{BS_j}(\bold{x}). \label{eq:T_cluster}
\end{equation}
\begin{corollary} \Mk{(MGF of $T_{cluster}(\bold{x})$).}
    For a cluster located at point $\bold{x}\in \psi$ with $N \sim \textsf{Pois}(\lambda_{BS})$ BSs equidistant (with distance $x=\lVert \bold{x}-\bold{h}\lVert$) to the satellite, the MGF of \eqref{eq:T_cluster} is:
    \begin{multline}
    M_{cluster}(\eta;\; x, \;g) = \\ \exp \left( 
        \lambda_{BS}\left(-1+M_{BS}(\eta;\; x, \;g)  \right)
    \right). \label{eq:MGF_Cluster}
    \end{multline}
\end{corollary}
\begin{proof}
    Since $N$ is a Poisson random variable, the MGF of $N$ is $M_{N}(\eta)=\exp \left(\lambda_{BS}(e^\eta-1) \right)$. By setting $\eta:=g\omega^{\alpha}x^{-\alpha} \eta$, we acquire \eqref{eq:MGF_Cluster}.
\end{proof}

$T_{cluster}(\bold{x})$ in \eqref{eq:T_cluster} is the fundamental unit of RFI brightness temperature in our model. We obtain its series expansion to facilitate the calculation of RFI brightness temperature cumulants later on. We note that the series expansion of \eqref{eq:MGF_Cluster} will be of the general form:
\begin{equation}
    M_{cluster}(\eta;\; x, \;g) = \sum\nolimits_{n=0}^{\infty}p_n(\lambda_{BS})(g \omega x^{-\alpha})^n \; \frac{\eta^n}{n!}, \label{eq:SE_cluster}
\end{equation}
where $p_n(\lambda_{BS})$ depends on the channel model (i.e., Rayleigh or SRF), as shown in the following two lemmas.

\begin{lemma} \label{Lemma:p_n_rayleigh}
    Under Rayleigh fading, $p_n(\lambda_{BS})$ in \eqref{eq:SE_cluster} can be expressed as: 
    \begin{equation}
        p_n(\lambda_{BS}) = (2b_0)^n \Gamma(n) L_{n}^{(-1)}(-\lambda_{BS}), \label{eq:p_n_rayleigh}
        \end{equation}
    where 
\begin{equation}
    L_{n}^{(a)}(v)=\sum\nolimits_{i=0}^{n} (-1)^{i} \binom{n+a}{n-i}\frac{v^i}{i!}
\end{equation}
are the generalized Laguere polynomials.
\end{lemma}
\begin{proof}
Refer to Appendix \ref{Appendix:p_n_rayleigh}.    
\end{proof}

\begin{lemma}\label{Lemma:p_n_SRF}
    Under shadowed Rician fading, for $n \geq 1$, $p_n(\lambda_{BS})$ in \eqref{eq:SE_cluster} can be expressed as:
    \begin{equation}
        p_n(\lambda_{BS}) = \sum\nolimits_{i=1}^{n} B_{n,i}(a_1,..., a_{n+i-1})\lambda_{BS}^i, \label{eq:p_n_SRF}
    \end{equation}
    where $B_{n,i}(\cdot)$ is the impartial Bell polynomial of order $n$, and $a_j$ is the $j$th moment of channel gain $|H|^2$ defined in \eqref{eq:A_n}.
\end{lemma}
\begin{proof}
Refer to Appendix \ref{Appendix:p_n_SRF}.    
\end{proof}

\MKK{For $n \in \{1,2 \}$, $p_{n}(\lambda_{BS})$ is required for the first and second cumulants and, consequently, the average and variance of RFI brightness temperature. Since we focus on these in our results (Section~\ref{sec:results}), we acquire $p_{1}(\lambda_{BS})$ and $p_{2}(\lambda_{BS})$ under Rayleigh and SRF in the following two corollaries.}

\begin{corollary} Under Rayleigh fading in \eqref{eq:p_n_rayleigh}:
    \begin{equation}
        p_{1}(\lambda_{BS}) = (2b_0)\lambda_{BS} \label{eq:p_1_Rayleigh}
    \end{equation}
    and under shadowed Rician fading in \eqref{eq:p_n_SRF}:
    \begin{equation}
        p_{1}(\lambda_{BS}) = (2b_0)\lambda_{BS} \left( m\frac{2b_0m +\Omega}{2b_{0}m} -m +1 \right). \label{eq:p_1_SRF}
    \end{equation}
\end{corollary}

\begin{corollary} Under Rayleigh fading in \eqref{eq:p_n_rayleigh}:
    \begin{equation}
        p_{2}(\lambda_{BS}) = (2b_0)^2(\lambda_{BS}^2 +2\lambda_{BS}) \label{eq:p_2_Rayleigh}
    \end{equation}
    and under shadowed Rician fading in \eqref{eq:p_n_SRF}:
\begin{align}
    p_{2}(\lambda_{BS}) & = ~ (2b_0)^{2} \left( 
    \lambda_{BS}^2  \left( m\frac{2b_0m +\Omega}{2b_{0}m} -m +1 \right)^2 \right. \notag \\
    & + \left. \lambda_{BS} \left( 
    m(m+1) \left( \frac{2b_0m +\Omega}{2b_{0}m} \right)^2 \right. \right. \notag \\
    & \left. \left. -2m(m-1) \frac{2b_0m +\Omega}{2b_{0}m} +(m-1)(m-2)
    \right)
    \right) . \label{eq:p_2_SRF}
\end{align}   
\end{corollary}

\subsection{MGF of RFI on SMAP's Main- and Side-lobes}

\subsubsection{SMAP's main-lobe} The RFI brightness temperature $T_{(ml)}$ on SMAP's main-lobe is caused by all the clusters $\bold{X} \in \Psi_{(ml)}$ in SMAP's main-lobe antenna footprint: i.e., 
\begin{equation}
    T_{(ml)}= \sum\nolimits_{\bold{X}_i \in \Psi_{(ml)}} T_{cluster}(\bold{X}_i). \label{eq:T_RFI_ml}
\end{equation}
\begin{note}
    Given that SMAP's antenna footprint is relatively small compared to the distance to the satellite, we assume that all cluster centers within the main-lobe antenna footprint ($\bold{X} \in \Psi_{(ml)}$) are approximately equidistant from the satellite. Under this assumption, the distance to the satellite for the clusters in the main-lobe, $d_{(ml)}$, is the answer to the following quadratic equation:
    \begin{equation}
      d_{(ml)}^2 + 2d_{(ml)}r_e\sin(40\degree)+r_e^2-h^2=0,
    \label{eq:ML_Dist_Sat}
\end{equation}
where $40\degree$ is the incident angle of SMAP's antenna on Earth.
\end{note}

\begin{lemma} \label{Lemma:MGF_RFI_ml} \Mk{(MGF of RFI BT on main-lobe $T_{(ml)}$).}
With the assumption of $M \sim \textsf{Pois}(\Lambda)$ equidistant  clusters (with distance $d_{(ml)}$) to the satellite located in SMAP's main-lobe antenna footprint $\mathcal{B}_{(ml)}$, where $\Lambda=\lambda_c\text{v}^2(\mathcal{B}_{(ml)})$ and $\text{v}^2(\cdot)$ is a Lebesgue measure in $\mathbb{R}^2$, the MGF of $T_{(ml)}$ defined in \eqref{eq:T_RFI_ml} is as follows:
\begin{multline}
    M_{(ml)}(\eta)= \\ \exp \left(40^{2}\lambda_{c} \left(-1+M_{cluster}(\eta; \; d_{(ml)}, \; g_{(ml)}) \right) \right), \label{eq:MGF_Main_Lobe}
\end{multline}
where $M_{cluster}(\eta; \; x,\; g)$ is defined in \eqref{eq:MGF_Cluster}.
\end{lemma}
\begin{proof}
Refer to Appendix \ref{Appendix:MGF_RFI_ml}.
\end{proof}

\subsubsection{SMAP's side-lobe} In this section, we investigate the 
RFI brightness temperature $T_{(sl)}$ on SMAP's side-lobe defined in \eqref{eq:T_RFI}, which we can rewrite as:
\begin{equation}
    T_{(sl)}=\sum\nolimits_{\bold{X}_i \in \Psi_{(sl)}} T_{cluster}(\bold{X}_i), \label{eq:T_RFI_sl}
\end{equation}
where $T_{cluster}(\bold{X}_i)$ is defined in \eqref{eq:T_cluster}.

\begin{lemma} \label{Lemma:MGF_RFI_sl}
    \Mk{(MGF of RFI BT on side-lobe $T_{(sl)}$).} The MGF of \eqref{eq:T_RFI_sl} is as follows:
    \begin{align}
       & \hspace{-4mm}M_{(sl)}(\eta)=\nonumber \\
       & \hspace{-5mm} \exp \hspace{-.8mm}{ \left(\hspace{-.6mm}-2\pi\hspace{-.6mm}\left(\frac{r_e}{h}\right)\hspace{-.6mm} \lambda_c\hspace{-.6mm} \int_{d_{min}}^{d_{max}\hspace{-.6mm}}\hspace{-.6mm}\hspace{-.6mm} \left(1-M_{cluster}(\eta; \;x,\;g_{(sl)}) \right)x\,dx \hspace{-.6mm}\right)} \label{eq:MGF_Side_Lobe}\hspace{-5mm}
    \end{align}
    where $M_{cluster}(\eta; \;x, \; g)$ is defined in \eqref{eq:MGF_Cluster}, and $d_{min}=h-r_e$ and $d_{max}=\sqrt{h^2 - r^2_e}$.
\end{lemma}
\begin{proof}  
Refer to Appendix \ref{Appendix:MGF_RFI_sl}.
\end{proof}

\subsection{Cumulants of RFI Brightness Temperature on SMAP's Main- and Side-lobes}
Now that we have the MGFs of $T_{(ml)}$ and $T_{(sl)}$, we are able to acquire their cumulants. 

\begin{lemma} \Mk{(Cumulants of RFI BT on main-lobe).} The $n$th cumulant of RFI brightness temperature $T_{(ml)}$ on SMAP's main-lobe,  defined in \eqref{eq:T_RFI_ml}, is as follows: 
\begin{equation}
     k_{n}^{(ml)}=40^{2} g_{(ml)}^n \omega^{n} \lambda_cp_n(\lambda_{BS})d_{(ml)}^{-n\alpha}, \label{eq:kn_ml}
\end{equation}
where $p_{n}(\lambda_{BS})$ is defined in \eqref{eq:p_n_rayleigh} and \eqref{eq:p_n_SRF} for Rayleigh and SRF models, respectively.
\end{lemma}
\begin{proof}
The cumulants of $T_{(ml)}$ can be acquired using its MGF as defined in \eqref{eq:MGF_Main_Lobe} and Definition 1.
\end{proof}

\begin{corollary} \Mk{(Average of RFI BT on main-lobe).} \label{cor:avg_T_ml} 
The expected value of $T_{(ml)}$ defined in \eqref{eq:T_RFI_ml} is:
\begin{equation}
\mathds{E}[T_{(ml)}]=40^2 g_{(ml)} \omega \lambda_c p_{1}(\lambda_{BS}) d_{(ml)}^{-\alpha},   \label{eq:T_ml_average} 
\end{equation}
where $p_{1}(\lambda_{BS})$ is defined in \eqref{eq:p_1_Rayleigh} and \eqref{eq:p_1_SRF} for Rayleigh and SRF models, respectively.
\end{corollary}
\begin{proof}
Based on Remark 1, the expected value of $T_{(ml)}$ is its first cumulant defined in \eqref{eq:kn_ml} with $n=1$.
\end{proof}

\begin{corollary} \Mk{(Variance of RFI BT on main-lobe).}\label{cor:STD_T_ml}
The variance of $T_{(ml)}$ defined in \eqref{eq:T_RFI_ml} is:
\begin{equation}
Var[T_{(ml)}]=40^2g_{(ml)}^2 \omega^{2} \lambda_c p_{2}(\lambda_{BS})d_{(ml)}^{-2\alpha},    \label{eq:T_ml_STD}
\end{equation}
where $p_{2}(\lambda_{BS})$ is defined in \eqref{eq:p_2_Rayleigh} and \eqref{eq:p_2_SRF} for Rayleigh and SRF models, respectively.
\end{corollary}
\begin{proof}
Based on Remark 1, the variance of $T_{(ml)}$ is its second cumulant defined in \eqref{eq:kn_ml} with $n=2$.
\end{proof}

\begin{lemma} \Mk{(Cumulants of RFI BT on side-lobe).} 
    The $n$th cumulant of the RFI brightness temperature $T_{(sl)}$ on SMAP's side-lobe, defined in \eqref{eq:T_RFI_sl}, is as follows:
    \begin{multline}
        k_n^{(sl)}= \\ \frac{2\pi}{2-n\alpha}\left(\frac{r_e}{h}\right) g_{(sl)}^n \omega^{n} \lambda_c p_n(\lambda_{BS}) \left(d_{max}^{2-n\alpha} - d_{min}^{2-n\alpha}\right) \label{eq:kn_sl},
    \end{multline}
where $p_{n}(\lambda_{BS})$ is defined in \eqref{eq:p_n_rayleigh} and \eqref{eq:p_n_SRF} for Rayleigh and SRF models, respectively.
\end{lemma}
\begin{proof}
The cumulants of $T_{(sl)}$ can be acquired using its MGF as defined in \eqref{eq:MGF_Side_Lobe} and Definition 1.
\end{proof}

\begin{corollary}\Mk{(Average of RFI BT on side-lobe).} \label{cor:avg_T_sl} 
The expected value of $T_{(sl)}$ defined in \eqref{eq:T_RFI_sl} is:
    \begin{multline}
        \mathds{E}[T_{(sl)}]= \\ \frac{2\pi}{2-\alpha}\left(\frac{r_e}{h}\right) g_{(sl)} \omega \lambda_c p_1(\lambda_{BS}) \left(d_{max}^{2-\alpha} - d_{min}^{2-\alpha}\right), \label{eq:T_sl_average}
    \end{multline}
    where $p_{1}(\lambda_{BS})$ is defined in \eqref{eq:p_1_Rayleigh} and \eqref{eq:p_1_SRF} for Rayleigh and SRF models, respectively.
\end{corollary}
\begin{proof}
Based on Remark 1, the expected value of $T_{(sl)}$ is its first cumulant defined in \eqref{eq:kn_sl} with $n=1$.
\end{proof}

\begin{corollary} \Mk{(Variance of RFI BT on side-lobe).} \label{cor:STD_T_sl} 
The variance of $T_{(sl)}$ defined in \eqref{eq:T_RFI_sl} is:
    \begin{multline}
        Var[T_{(sl)}]= \\ \frac{2\pi}{2-2\alpha}\left(\frac{r_e}{h}\right) g_{(sl)}^{2} \omega^{2} \lambda_c p_2(\lambda_{BS}) \left(d_{max}^{2-2\alpha} - d_{min}^{2-2\alpha}\right), \label{eq:T_sl_STD}
    \end{multline}
    where $p_{2}(\lambda_{BS})$ is defined in \eqref{eq:p_2_Rayleigh} and \eqref{eq:p_2_SRF} for Rayleigh and SRF models, respectively.
\end{corollary}
\begin{proof}
Based on Remark 1, the variance of $T_{(sl)}$ is its second cumulant defined in \eqref{eq:kn_sl} with $n=2$.
\end{proof}

\nick{\subsection{Accounting for Heterogeneous Clusters}}
In the previous section, we derived cumulants for a homogeneous set of parameters, including the intensity of BSs per cluster $\lambda_{BS}$, the BS transmission power $p_{tx}$, and SRF parameters $(\Omega, \;2b_0, \;m )$. However, in practice, these values may vary across different clusters or even among BSs within a cluster. For instance, $\lambda_{BS}$ depends on the area and population density of a cluster. Similarly, SRF parameters might be influenced by factors such as the proximity of a BS to the cluster center. Moreover, these parameters can exhibit correlations; for example, in clusters with a higher $\lambda_{BS}$, the transmission power $p_{tx}$ of the BSs might be reduced to mitigate co-channel inteference \nick{among BSs in the cluster}.

To address this heterogeneity and enhance the generality of our model, we propose using conditional cumulants. By considering the parameter set $\Theta = \{\lambda_{BS}, p_{tx}, (\Omega, \; 2b_0, \; m) \}$, we can calculate the conditional cumulants $k_{n|\Theta}^{(ml)}$ and $k_{n|\Theta}^{(sl)}$. Consequently, the average cumulants for the main- and side-lobe RFI brightness temperatures can be determined as:
\begin{equation}
    k_n^{(l)} = \mathds{E}_{\Theta}[k_{n|\Theta}^{(l)}]
\end{equation}
\nick{where $(l)$ is either $(ml)$ or $(sl)$. However, due to space limitations, in this paper we focus our results on homogeneous values of parameters across all BS clusters.}

\section{Spectral Efficiency and Throughput}\label{sec:spectral_efficiency}
In the last section, we investigated the impact of RFI from the downlink of a large-scale NextG terrestrial network on SMAP, and how to mitigate its impact. In this section, we evaluate the acquired average \textit{spectral efficiency} ($s_e$) within each cluster of BSs and the average sum throughput of the terrestrial network. 

\begin{note}
  We assume that clusters are far from each other. Thus, for evaluating $s_e$ in one cluster, we neglect the impact of RFI from neighboring clusters. Also, within a cluster, we neglect the curvature of the Earth.
\end{note}
In a 2D scenario, we imagine the cluster center is located at $O=(0,0)$, and there are $N \sim \textsf{Pois}{(\lambda_{BS})}$ BSs in the cluster. Then, according to the definition of the TCP, a BS's location $\bold{Y} \Mk{ \in \mathbb{R}^2} $ with respect to the cluster's center has distribution $g_{\bold{Y}}(\bold{y}) = \frac{1}{2\pi\sigma_{c}^2} \exp{\left\{ -\frac{\lVert \bold{y}\lVert^{2}}{2\sigma_c^2}\right\}}$. We also assume that each UE's location $\bold{Z} \Mk{ \in \mathbb{R}^2}$ is similarly distributed with respect to the center of the cluster, i.e., $g_\bold{Z}(\bold{z}) = \frac{1}{2\pi\sigma_{c}^2} \exp{\left\{ -\frac{\lVert \bold{z}\lVert^{2}}{2\sigma_c^2}\right\}}$. \MKK{We assume that the BS closest to a UE acts as its serving BS.}

\begin{remark}
    According to \cite[Corollary 1]{afshang2018poisson}, the distance $Z=\lVert \bold{Z} \lVert$ of a UE to the cluster center follows a \MKK{Rayleigh distribution}, i.e., $f_Z(z)=\frac{z}{\sigma_c^2}\exp{\left( - \frac{z^2}{2\sigma_c^2} \right)}$ with \MKK{scale parameter $\sigma_c$}.
    \MKK{Accordingly, the distribution of the distance $R$ of each (randomly selected) BS, positioned with distribution $g_{\bold{Y}}$ within the cluster, from a UE with distance $Z=z$ from the cluster center, i.e., $R = \lVert \bold{Y} - \bold{z} \lVert$, is Rician: i.e.,}
    \begin{align}
        f_{R}(r|z) &= \textsf{Ricepdf}(r,z;\sigma_c) \notag \\
        &= \frac{r}{\sigma_c^2} \exp{ \left( - \frac{r^2+z^2}{2\sigma_c^2} \right) } I_{0} \left( \frac{rz}{\sigma_c^2}\right),
    \end{align}
    where $I_{0}(\cdot)$ is the zeroth order modified Bessel function. Also, we denote $F_{R}(r|z)$ as the CDF of $R$, which is equivalent to:
    \begin{equation}
        F_{R}(r|z) = 1 - Q_1\left(\frac{z}{\sigma_c}, \frac{r}{\sigma_c}\right),
    \end{equation}
    where $Q_1(\cdot)$ is the Marcum Q-function. Conditioned on $z$, the PDF of the distance of the first \textit{serving} BS to a UE is:
    \begin{equation}
        f_{serv}(r|z)=\lambda_{BS} f_R(r|z) \exp\left(-\lambda_{BS} F_R(r|z)\right).
    \end{equation}
\end{remark}

To acquire $s_e$ within a cluster, we take a similar approach to \cite{deshpande2020spectral}. We assume a Rayleigh channel $H_{\Mk{i}} \sim \textsf{Exp}(1)$ between the UE and a BS \Mk{($i$)}, and define the following parameters:
\begin{align}
    \text{SINR} &= \frac{P_s}{\Mk{\beta}N_0+I} \notag \\
    P_s & = \MK{ u }  p_{tx} H_{\Mk{i}} R_i^{-\alpha}, \notag \\
    I & = \MK{ u } \sum\nolimits_{N \big\backslash BS_{serv}} p_{tx} H_{\Mk{i}} R_i^{-\alpha}
\end{align}
where \MK{$u =  \left( \frac{c}{4\pi f} \right)^2 $}, $P_s$ is the power received by the UE \Mk{from the serving BS}, $I$ is the interference caused by the other cells, $N_0$ is the noise power in the environment, and \Mk{$R_i$ is the distance from BS ($i$) to the UE}. Accordingly, $s_e$ can be defined as:
\begin{equation}
    s_e = \mathds{E}[\ln{(1+\text{SINR})}] 
\end{equation}
Similar to \cite{deshpande2020spectral}, we use Hamdi's lemma to acquire $s_e$:
\begin{equation}
    s_e= \int_{0}^{\infty} \frac{\exp(-sN_0)}{s} \left ( \mathcal{L}_I(s)-\mathcal{L}_P(s) \right)ds, \label{eq:s_e}
\end{equation}
where $\mathcal{L}_I(s)$ and $\mathcal{L}_P(s)$ denote the Laplace Transforms (LTs) of interference $I$ and total received power $P=I+P_s$, respectively. In Lemmas \ref{Lemma:Laplace_P} and \ref{Lemma:Laplace_I}, we derive $\mathcal{L}_P(s)$ and $\mathcal{L}_I(s)$.

\begin{lemma} \label{Lemma:Laplace_P} 
In \eqref{eq:s_e}, $\mathcal{L}_P(s)$ can be expressed as:
    \begin{align}
        &\mathcal{L}_P(s) = \mathds{E}_Z \left[ \mathcal{L}_P(s|Z) \right] \notag \\
        & = \int_{0}^{\infty} f_Z (z) \exp \left(-\lambda_{BS} \left(1-\int_0^\infty \frac{f_R(r|z)}{1+(\MK{u}p_{tx}r^{-\alpha})s}dr \right) \right) dz. \label{eq:Laplace_P}
     \end{align}
\end{lemma}
\begin{proof}
    Refer to Appendix \ref{Appendix:Laplace_P}.
\end{proof}

\begin{lemma}
\label{Lemma:Laplace_I} 
In \eqref{eq:s_e}, $\mathcal{L}_I(s)$ can be expressed as:
    \begin{align}
        \mathcal{L}_I(s) &=\mathds{E} [\mathcal{L}_I(s|Z)] \MK{+ \exp(-\lambda_{BS}) }, \notag 
    \end{align}
    with:
    \begin{equation}
        \mathds{E} [\mathcal{L}_I(s|Z)] = \int_{0}^{\infty} f_Z(z) \mathcal{L}_I(s|z) \; dz, 
    \end{equation}
    where:
    \begin{align}
        \mathcal{L}_I(s|z)=&\int_{0}^{\infty}  f_{serv}(q|z) \times \notag \\  &\exp \left(-\lambda_{BS} \int_q     ^\infty \frac{(\MK{u}p_{tx}r^{-\alpha})s}{1+(\MK{u}p_{tx}r^{-\alpha})s}f_R(r|z)\;dr  \right) dq.     
    \end{align}
\end{lemma}
\begin{proof}
Here, the distribution of the first serving BS also matters. The rest of the proof is similar to Lemma \ref{Lemma:Laplace_P}. \MK{Note that the term $\exp(-\lambda_{BS})$ is for the probability that there is no serving BS in the cluster (i.e., there is no BS in the cluster). If this happens, then the Laplace of $I$ is equal to 1.}
\end{proof}

After deriving $s_e$, we can use the following formula to acquire the average sum throughput of the terrestrial network (including all clusters and BSs):
\begin{equation}
    t_p = \beta s_e\lambda_{BS}\lambda_c\left(2\pi r_e^2\left(1-\frac{r_e}{h}\right)\right), \label{eq:t_p}
\end{equation}
where $\beta$ is the system bandwidth.



\section{Simulation Results}\label{sec:results}

\begin{table}[][ht]
\centering
\caption{Simulation parameters}
\label{tab:SG_Params}
\begin{tabular}{|c|c|}
\hline
\textbf{Element} & \textbf{Value}                                                                              \\ \hline
Intensity of clusters ($\lambda_c$)      & \begin{tabular}[c]{@{}c@{}}$1$ cluster (dense urban area) \\ every 10,000 km$^{2}$\end{tabular}  \\ \hline
Intensity of \nick{BSs per cluster} ($\lambda_{BS}$)   & \begin{tabular}[c]{@{}c@{}}$500$, $800$, $1200$\\ BSs per cluster\end{tabular} \\ \hline
Earth-space path loss exponent ($\alpha$)         & {(}2 , 2.4{]}                                                                               \\ \hline
BS transmission power ($p_{tx}$)                & $20$ Watts (43 dBm)   \\ \hline
Channel fading params. ($b_0$,$m$,$\Omega$)        & ($0.158$, $0.739$, $8.97\times10^{-4}$)                                                                            \\ \hline
Boltzmann's constant ($k_b$) & \begin{tabular}[c]{@{}c@{}} $1.380649 \times 10^{-23}$\\ m$^{2}$kg s$^{-2}$K$^{-1}$ \end{tabular} \\ \hline
Speed of light ($c$) & 300,000 km s$^{-1}$ \\ \hline
SMAP main-lobe antenna gain ($g_{(ml)}$)            & $0$ dB                                                                                    \\ \hline
SMAP side-lobe antenna gain ($g_{(sl)}$)            & $-50$ dB                                                                                    \\ \hline
BS central carrier frequency ($f$)                & $1.413$ GHz                                                                                                                                                           \\ \hline
BS transmission bandwidth ($\beta$)                & $24$ MHz                                                                                    \\ \hline
Earth radius ($r_e$)                & $6371$ km                                                                                   \\ \hline
SMAP's distance to Earth's center ($h$)                & $7056$ km                                                                                   \\ \hline
Min. distance to SMAP ($d_{\min}$)        & $h-r_e=685$ km                                                                                    \\ \hline
Max. distance to SMAP ($d_{\max}$)        & $\sqrt{h^2-r_e^2}=3032.7$ km                                                                            \\ \hline
Main-lobe distance to SMAP ($d_{(ml)}$)        &  $865.5$ km                                                                            \\ \hline
Clusters' \MKK{scale parameter} ($\sigma_c$)        &  $4$ km                                                                            \\ \hline
\MK{UE noise power density ($N_0$)}        & \MK{ $-174$ dBm/Hz}                                                                            \\ \hline
\MK{\nick{Intra}-cluster path loss exponent ($\alpha$)}        & \MK{4}                                                                            \\ \hline
\end{tabular}
\end{table}

In this section, we evaluate the RFI brightness temperature statistics at both SMAP's main- and side-lobes based on the analysis in the previous sections. 
The simulation parameters are given in Table \ref{tab:SG_Params}. Based on this table, an average of one BS cluster exists in every 10,000 km$^2$ and, on average, there are $2500$ clusters in the area exposed to the satellite.
We set the average number of BSs per cluster to $\lambda_{BS} \in \{500, 800, 1200\}$, which corresponds to averages of $1.25$ million, $2$ million, and $3$ million cellular BSs exposed to SMAP, respectively. 
We allocate a transmission bandwidth of $24$ MHz (centered at 1.413 GHz) to each BS, complemented by a $3$ MHz guard-band, resulting in a total of $27$ MHz within the restricted L-band spectrum (1.400 to 1.427 GHz).
We set SMAP's side-lobe gain to a conservative $-50$ dB, while the main-lobe gain is $0$ dB. For the shadowed Rician channel model, we adopt $b_{0} = 0.126$ from the ``Average shadowing'' scenario in \cite[Table III]{abdi2003new}. As discussed in Section \ref{sec:Channel_Model}, we expect that the LoS RFI component stems only from the higher vertical side-lobes of a BS's antenna. We also expect rather heavy blockages in a dense urban area. For these reasons, we choose the values $m=0.739$ and $\Omega=8.97\times10^{-4}$ from the ``Frequent heavy shadowing'' scenario in \cite[Table III]{abdi2003new}.\footnote{\MK{We expect the LoS RFI to be weaker than the non-LoS RFI \MKK{because the LoS RFI only stems from the higher vertical side-lobes of a BS, while the Rayleigh component stems from the main-lobe.} Thus we choose this set of parameters. A more precise set of parameters for the SRF model between a BS and SMAP requires further study. 
\nick{However, the overall trends identified in our results will still hold for different  parameters.}
}}

\subsection{Average and Standard Deviation of RFI Brightness Temperature on SMAP's main- and side-lobes}\label{sec:avg_std_RFI_results}
In this section, we evaluate the average and standard deviation \MKK{(STD)} of RFI brightness temperature from the network described in Table \ref{tab:SG_Params} exposed to both SMAP's main- and side-lobes. 
The results are shown in Figure \ref{fig:ml_sl_avg_std}. The solid lines show the theoretical results while the markers ($\times$) show the Monte Carlo simulation results. 
\MKK{The Monte Carlo simulations are conducted over 10,000 rounds. In each round, the positions of clusters are generated randomly and the number of BSs in each cluster is determined based on a Poisson distribution with $\lambda_{BS} \in \{500, 800, 1200\}$. 
\Mk{Within each cluster, the distance $r$ of a BS from the cluster center is generated as a Rayleigh random variable (as in Remark 2) with a scale parameter \(\sigma_c = 4\). From the cluster center, the BS is placed on the Earth's surface by moving the distance $r$ from the center in a random direction uniformly distributed over \([0, 2\pi]\), tangent to the Earth's surface.
}
For each BS independently, Earth-space channel gain $H$ is generated.
For each value of $\lambda_{BS}$, the cluster RFI brightness temperature is calculated for seven different path loss exponents $\alpha \in (2, 2.4]$. The aggregate RFI for each combination of $\lambda_{BS}$ and $\alpha$ is obtained by summing the RFI brightness temperatures of all clusters within the round.}
\nick{Recall from Note~\ref{note:BS_locations} that our analysis of the RFI from one cluster uses the simplifying assumption that all BSs within it are located at its center. Importantly, our Monte Carlo simulations do not rely on this assumption. Instead, they use the true random realizations of BS locations within each cluster.}

\begin{figure*}[ht]
  \centering
    \centering
    \begin{subfigure}{0.47\textwidth}
      \centering
      \includegraphics[width=0.92\linewidth]{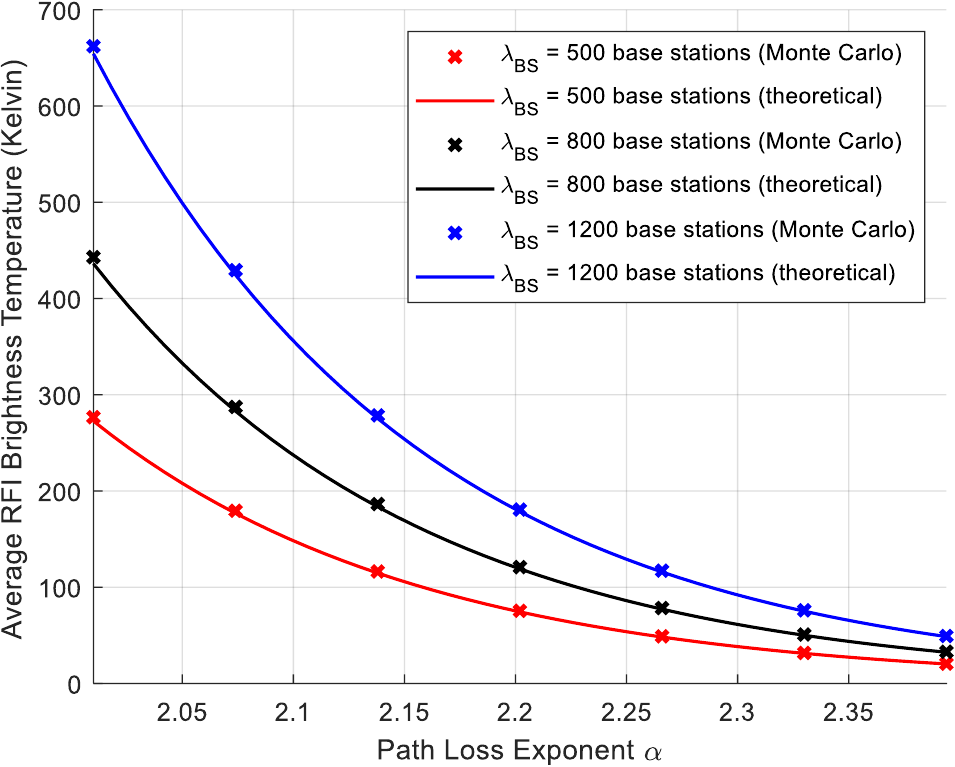}
      \caption{Main-lobe average RFI brightness temperature \MKK{(Corollary \ref{cor:avg_T_ml})}.}
      \label{fig:main_avg_RFI}
    \end{subfigure}
    \hspace{6 mm}
    \begin{subfigure}{0.47\textwidth}
      \centering
      \includegraphics[width=0.92\linewidth]{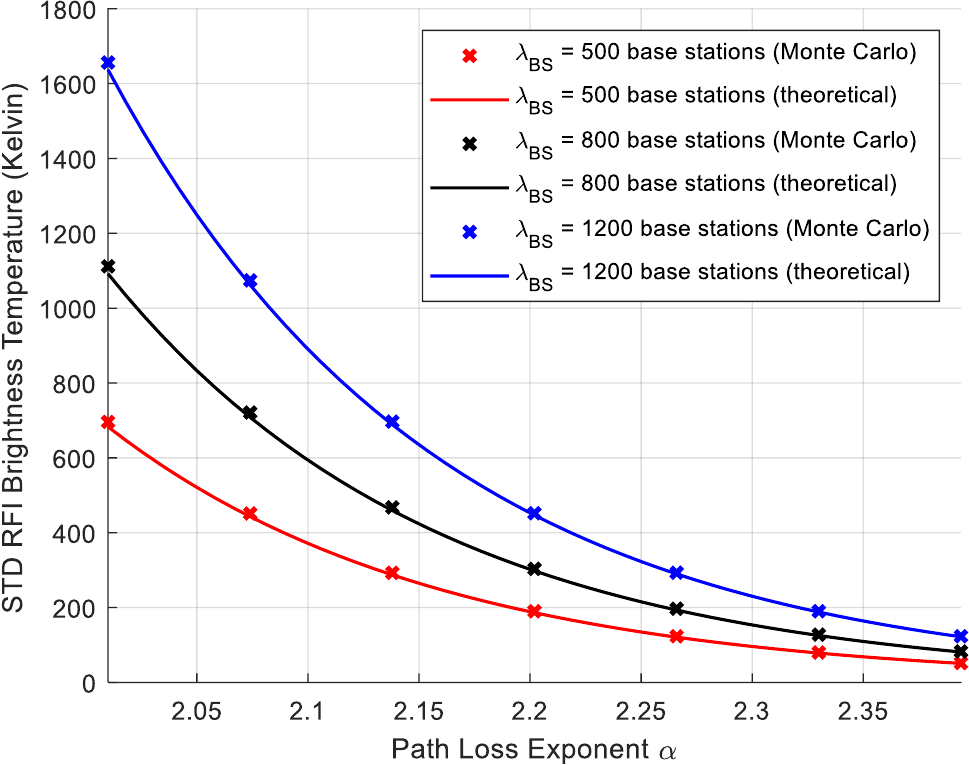}
      \caption{Main-lobe STD of RFI brightness temperature \MKK{(Corollary \ref{cor:STD_T_ml})}.}
      \label{fig:main_std_RFI}
    \end{subfigure}
    
    \centering
    \begin{subfigure}{0.47\textwidth}
      \centering
      \includegraphics[width=0.92\linewidth]{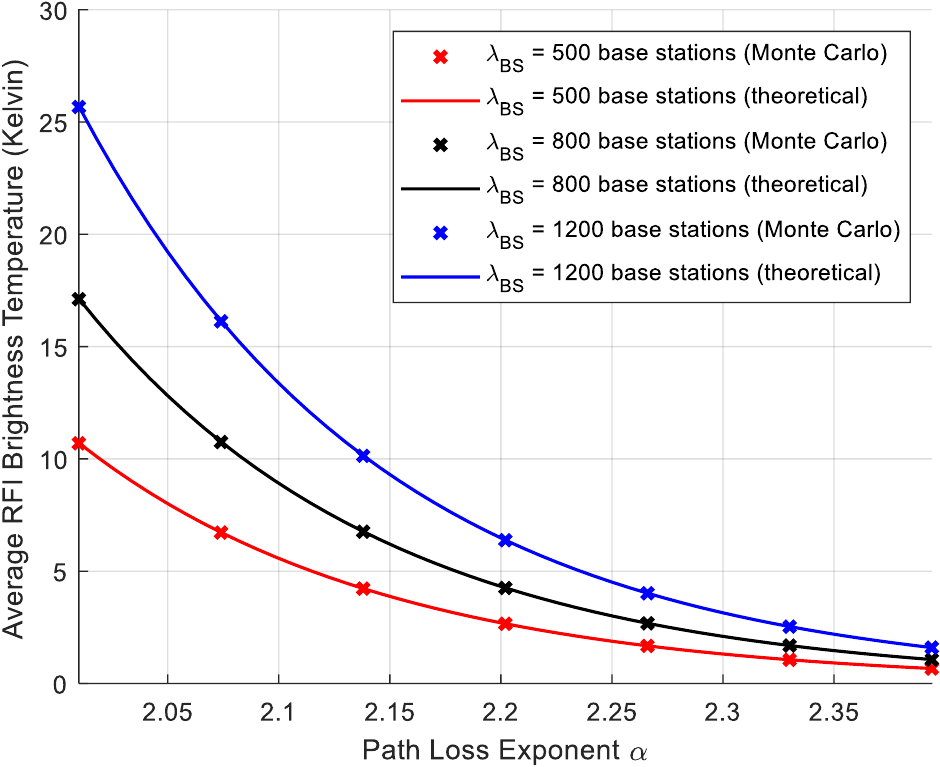}
      \caption{Side-lobe average RFI brightness temperature \MKK{(Corollary \ref{cor:avg_T_sl})}.}
      \label{fig:side_avg_RFI}
    \end{subfigure}
    \hspace{6 mm}
    \begin{subfigure}{0.47\textwidth}
      \centering
      \includegraphics[width=0.92\linewidth]{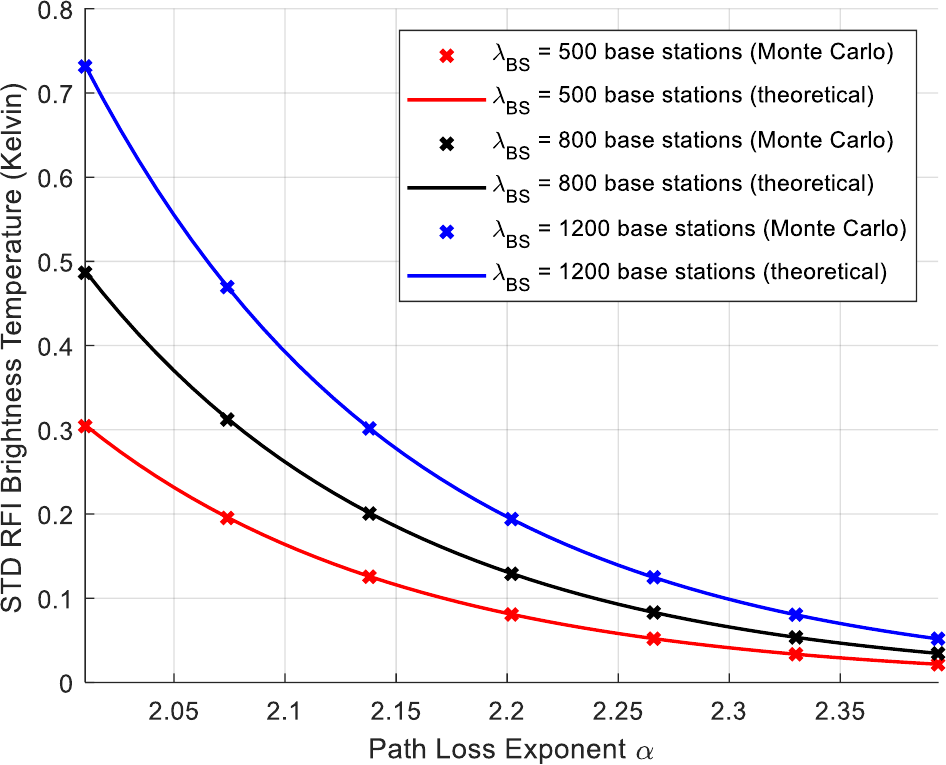}
      \caption{Side-lobe STD of RFI brightness temperature \MKK{(Corollary \ref{cor:STD_T_sl})}.}
      \label{fig:side_std_RFI}
    \end{subfigure}
  \caption{Average and STD of RFI brightness temperature on SMAP's main- and side-lobes for $\lambda_{BS}=500, \;800,\; 1200$ active BSs in a cluster. On average 1.25, 2, and 3 million BSs are exposed to SMAP's side-lobe while, on average, only $80$, $128$, and $\MKK{192}$ BSs are exposed to SMAP's main-lobe, respectively.}
  \label{fig:ml_sl_avg_std}
\end{figure*}
\captionsetup{skip=0\baselineskip}

As we can see from Figure \ref{fig:ml_sl_avg_std}, the average RFI brightness temperature on SMAP's main-lobe is on the order of hundreds of Kelvins (Figure \ref{fig:main_avg_RFI}) and the STD of RFI brightness temperature approaches close to 1,700 Kelvins (Figure \ref{fig:main_std_RFI}), which reveals the extreme stochasticity of the main-lobe RFI brightness temperature. It is notable that, on average, only $80$, $128$, and $192$ BSs contribute to these RFI characteristics on SMAP's main-lobe. \nick{Hence, as noted in Section~\ref{sec:mitigation},  protocols must be in place to prevent RFI on SMAP's main-lobe.}

Meanwhile, the average RFI brightness temperature on SMAP's side-lobe is on the order of tens of Kelvins (Figure \ref{fig:side_avg_RFI}). Note that on average $\MKK{1.25}$, $2$, and $3$ million BSs contribute to these RFI characteristics. 
\nick{Although there are significantly more active BSs exposed to SMAP's side-lobe than its main-lobe, the average side-lobe RFI is much lower than the main-lobe RFI due to SMAP's extremely low side-lobe gain.}

An especially noteworthy observation from Figure \ref{fig:ml_sl_avg_std} is the remarkably low STD of RFI brightness temperature on SMAP's side-lobe (Figure \ref{fig:side_std_RFI}), which consistently measures less than one Kelvin across all values of the path loss exponent $\alpha$ and cluster BS intensity $\lambda_{BS}$. These exceptionally low STD values suggest a minimal degree of variability in RFI brightness temperature, implying that the actual RFI brightness temperature realization is unlikely to deviate significantly from the average value by more than a few Kelvins. Consequently, there exists a high level of certainty in the RFI brightness temperature measurements. \nick{This suggests that, given an accurate estimate of the average side-lobe RFI, we should be able to achieve low SOPs as defined in \eqref{eq:SOP} using the RFI mitigation technique proposed in Section~\ref{sec:mitigation}.}

Finally, it is clear from Figure \ref{fig:ml_sl_avg_std} that the Monte Carlo simulations closely match the theoretical results. Thus, our assumption in Note \ref{note:BS_locations} that all BSs within a cluster are located at its center had a negligible impact on our theoretical results. 

\subsection{RFI Mitigation and Sensing Outage Probability (SOP)}

\begin{figure*}[ht]
  \centering
    
    \begin{subfigure}{0.47\textwidth}
      \centering
      \includegraphics[width=.95\linewidth]{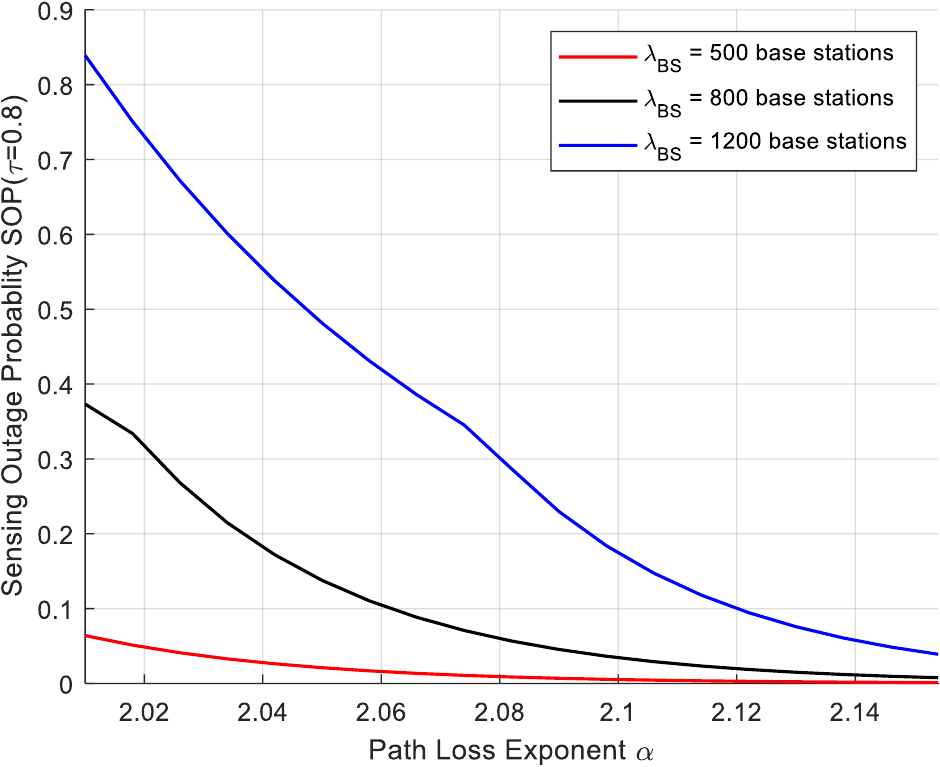}
      \caption{Upper-bound SOP for $\tau = 0.8$ Kelvin.}
      \label{fig:SOP_tau_0.8}
    \end{subfigure}
    \hspace{6 mm}
    \begin{subfigure}{0.47\textwidth}
      \centering
      \includegraphics[width=.95\linewidth]{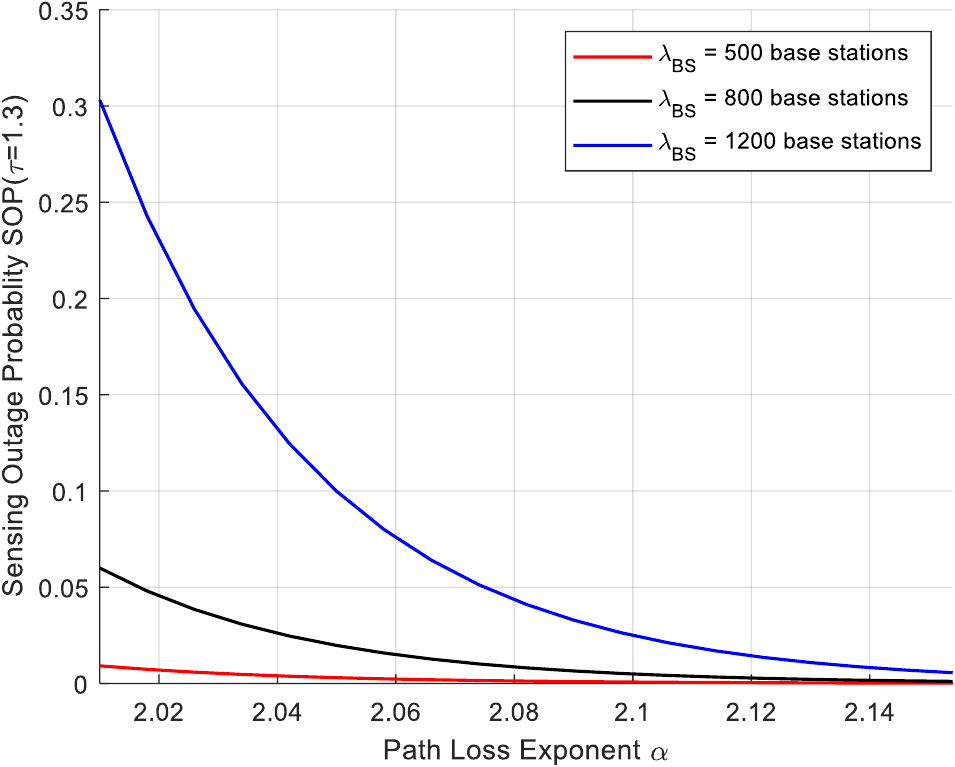}
      \caption{Upper-bound SOP for $\tau = 1.3$ Kelvin.}
      \label{fig:SOP_tau_1.3}
    \end{subfigure}

  \caption{Sensing outage probabilities $SOP(\tau)$ for $\tau=0.8$ and $\tau = 1.3$ Kelvin }
  \label{fig:Exc-Prob2}
\end{figure*}
\captionsetup{skip=0\baselineskip}

As we saw from Figure \ref{fig:ml_sl_avg_std}, 
\nick{the average RFI on the main- and side-lobes exceeds the tolerable threshold $\tau = 1.3$ Kelvin.}
To correct for this, we use the RFI mitigation technique introduced in Section \ref{sec:mitigation} and \eqref{eq:mitigation}. 
\nick{However, as noted in Section~\ref{sec:avg_std_RFI_results}, the RFI on SMAP's main-lobe exhibits extreme stochasticity, so the proposed RFI mitigation mechanism cannot be applied to correct for it. Instead, we focus on mitigating RFI on SMAP's side-lobe, for which the STD of RFI brightness temperature is less than 1 Kelvin (Figure~\ref{fig:side_std_RFI}), and evaluate the upper-bound on the SOP defined in \eqref{eq:SOP_up_bound}.}  

For the Chebyshev's inequality in \eqref{eq:SOP_up_bound}, we use the minimum of the $2$nd moment $\mu_2$ (variance) and the $4$th moment $\mu_4$ of RFI brightness temperature as:
\begin{equation}
    SOP(\tau)\leq \text{min} \left\{ \; \frac{\mu_2}{\tau^2} \;, \;\frac{\mu_4}{\tau^4} \; \right\}.
\end{equation}
\nick{We show the theoretical upper-bounds on the SOP for thresholds $\tau = 0.8$ and $\tau = 1.3$ Kelvin in Figures \ref{fig:SOP_tau_0.8} and \ref{fig:SOP_tau_1.3}, respectively. As we can see, for cluster BS intensity $\lambda_{BS} = 1200$, the SOP is almost always below $0.3$ at outage threshold $\tau = 1.3$ Kelvin. Meanwhile, for cluster BS intensities $\lambda_{BS} = 500$ and $800$, the SOP is almost always below $0.05$, which means that, on average, fewer than 5 out of every 100 samples taken by SMAP will be corrupted by RFI exceeding the tolerable threshold of $\tau = 1.3$ Kelvin. }


    


\nick{When we compare the upper-bounds on the SOPs for $\tau = 0.8$ and $\tau = 1.3$ Kelvin in Figures \ref{fig:SOP_tau_0.8} and \ref{fig:SOP_tau_1.3}, respectively, it is clear that the tighter threshold $\tau = 0.8$ results in higher upper-bounds. However, it is well-known that Chebyshev's inequality yields loose upper-bounds. To understand how loose these bounds are, we use Monte Carlo simulations to estimate the true SOP for $\tau = 0.8$.}
We compare the Monte Carlo simulations and theoretical results in Table \ref{tab:SOP}. The Monte Carlo results are acquired from 20,000 rounds of simulations for each value of $\lambda_{BS}$ and path loss exponent $\alpha$. \nick{The table clearly shows that the true SOP is far below the theoretical upper-bound. Therefore, we can safely use this loose upper-bound to ensure that $SOP(\tau)$ is itself below a suitably low threshold for the downstream scientific measurements/applications.} 

\begin{table*}[ht]
\centering
\caption{Sensing Outage Probability SOP($\tau=0.8~\Mk{\text{K}}$) for 20,000 rounds of simulation}
\label{tab:SOP}
\begin{tabular}{|c|c|c|c|c|c|c|c|}
\hline
$\lambda_{BS}$          & Method      & $\alpha=2.01$ & $\alpha =2.042$ & $\alpha =2.074$ & $\alpha =2.106$ & $\alpha =2.138$ & $\alpha =2.170$ \\ \hline
\multirow{2}{*}{$1200$} & Theoretical & 0.8391        & 0.5381          & 0.3452          & 0.1474          & 0.0608          & 0.0251          \\ \cline{2-8} 
                        & Monte Carlo & 0.2772        & 0.1763          & 0.0882          & 0.0327          & 0.0075          & 0.0008          \\ \hline
\multirow{2}{*}{$800$}  & Theoretical & 0.3733        & 0.1720          & 0.0708          & 0.0292          & 0.0120          & 0.0050          \\ \cline{2-8} 
                        & Monte Carlo & 0.0997        & 0.0406          & 0.0103          & 0.0011          & 0               & 0               \\ \hline
\multirow{2}{*}{$500$}  & Theoretical & 0.0640        & 0.0263          & 0.0108          & 0.0045          & 0.0018          & 0.0008          \\ \cline{2-8} 
                        & Monte Carlo & 0.0083        & 0.0009          & 0               & 0               & 0               & 0               \\ \hline
\end{tabular}
\end{table*}

\subsection{Sensing Outage Probability vs. Sum Throughput}
In the preceding two sections, we examined how the interference from the downlink of a large-scale terrestrial NextG cellular network affects SMAP. In this section, we evaluate the attainable sum throughput (over all BSs in all clusters) in the terrestrial network. Our analysis predominantly relies on the mathematical framework established in Section \ref{sec:spectral_efficiency}. However, our objective is not to optimize the configuration of BSs within clusters or the size of clusters to maximize the overall system throughput. \nick{Instead, we aim to illustrate the trade-off between the terrestrial network's sum throughput and SMAP's SOP.}

We take the \MKK{scale parameter} $\sigma_c = 4$ km. 
\Mk{Based on a Rayleigh distribution, 99 percent of BSs fall within the roughly $10$ km radius of the cluster center. Thus the urban area covered by each cluster is roughly $300$ km$^{2}$.}
Also, we assume an intra-cluster path loss exponent $\alpha = 4$ between UEs and BSs within a cluster. We acquire the spectral efficiency ($s_e$) within a cluster defined in \eqref{eq:s_e} and consequently the overall throughput of all clusters ($t_p$) defined in \eqref{eq:t_p}. We increase the average BSs within a cluster ($\lambda_{BS}$) from $0$ to $1500$ and we acquire the spectral efficiency $s_e$ and sum throughput $t_p$ for each value (assuming approximately 2,500 total clusters). For each value of $t_p$, we also acquire the theoretical upper-bound on the SOP for $\tau \in \{0.6, 0.8, 1, 1.3\}$ Kelvin. \nick{The results are shown in Figure \ref{fig:s_e} along with an inset plot of the spectral efficiency $s_e$ versus the cluster BS intensity $\lambda_{BS}$.} 
As we can see from the inset plot, $s_e$ 
\nick{saturates at} approximately $1.5$ Nats/Hz/second for cluster BS intensity $\lambda_{BS} \geq 150$. \MKK{The equivalent sum throughput $t_p$ when $\lambda_{BS} = 150$ is $13.6$ Tb/s as shown by the dashed vertical line in Figure \ref{fig:s_e}.} \nick{This implies that we are able to operate the network at maximum spectral efficiency while achieving a negligible SOP; however, given a maximum allowable SOP, the total number of BSs must kept below a threshold, which constrains the achievable sum throughput of a NextG network operating co-channel with SMAP.}

\begin{figure}[ht]
  \includegraphics[width=\linewidth]{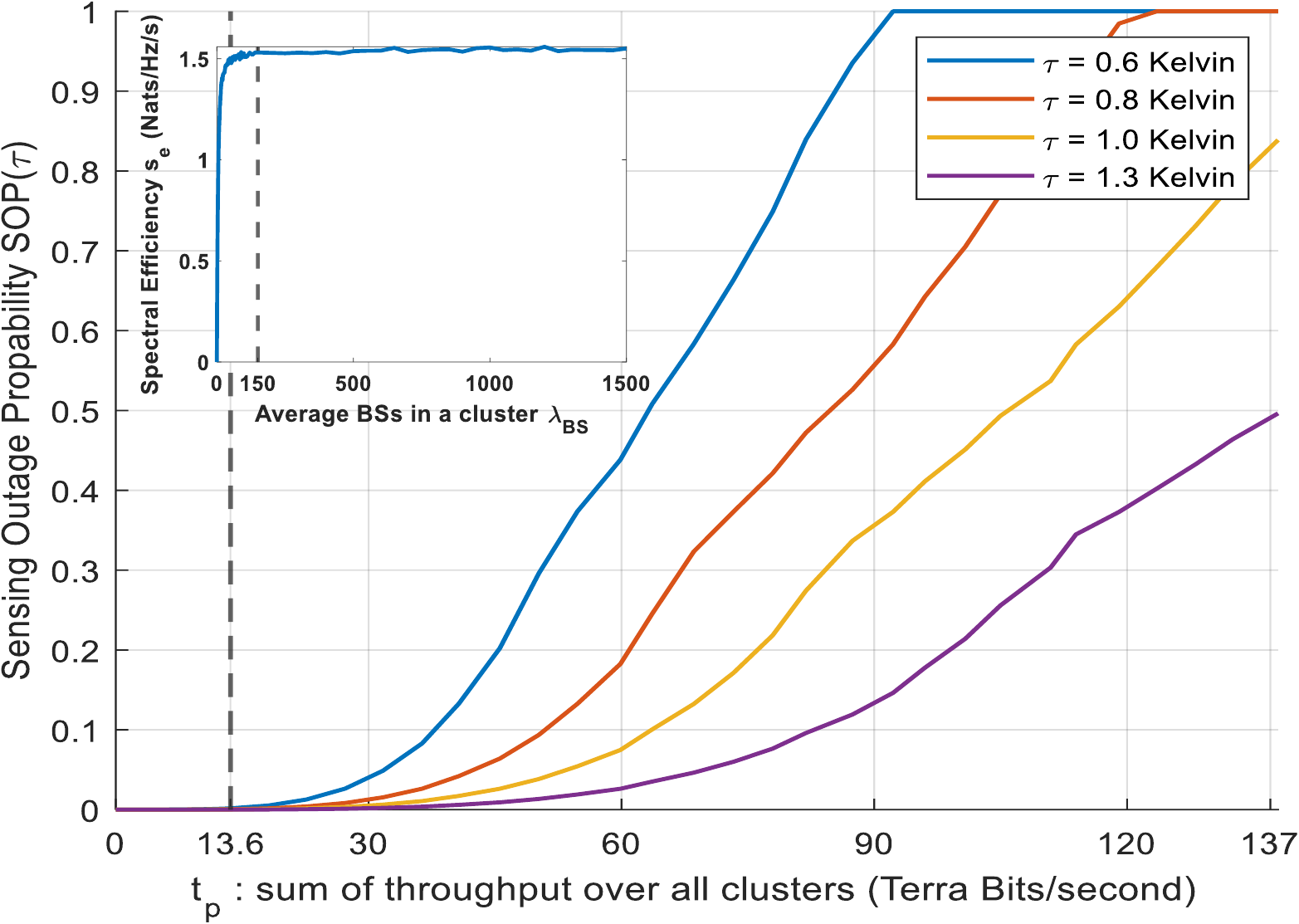}
  \caption{SOP for $\tau \in \{0.6, 0.8, 1, 1.3\}$ vs. sum throughput. \nick{Note that the throughput $t_p$ increases as the cluster BS intensity $\lambda_{BS}$ increases from $0$ to 1,500 over approximately 2,500 clusters. The inset plot shows the spectral efficiency $s_e$ within a cluster vs. the cluster BS intensity $\lambda_{BS}$.}
  }
  \label{fig:s_e}
\end{figure}

\section{Conclusion}\label{sec:conclusion}
In this paper, we proposed a model based on the Thomas cluster process from Stochastic geometry to evaluate RFI induced by a large-scale terrestrial network on an EESS satellite. We developed our model based on NASA's SMAP satellite, including its altitude and antenna characteristics.
We used the concept of cumulant and moment generating functions to acquire the statistical characteristics of RFI on both SMAP's main- and side-lobes, including the average, variance and the higher central moments. We validated our model using Monte Carlo simulations, and showed that a large number of active base stations can coexist while exposed to SMAP's side-lobe, while ensuring the accuracy of SMAP's measurements. We also acquired an upper-bound on the probability that RFI exceeds a tolerable threshold, which we called the Sensing Outage Probability (SOP). Lastly, we demonstrated a tradeoff between the SOP and the sum throughput of the terrestrial NextG network.


\appendices
\section{PROOF OF LEMMA \ref{lemma:SE_MGF_H}} \label{Appendix:SE_MGF_H}
Since $a_n$ is the $n$th moment of $\lvert H \lvert^2$, using the PDF of $\lvert H \lvert^2$ defined in \cite[Eq. (6)]{abdi2003new}, we note that:
\begin{multline}
    a_n
 = \frac{1}{2b_0}  \left( \frac{2b_{0}m}{2b_{0}m + \Omega} \right)^m \notag \\
 \int_{0}^{\infty} t^n \exp{\left(-\frac{t}{2b_0} \right)} \;_1F_1\left(m, 1, \frac{\Omega t}{2b_0(2b_0m+\Omega)} \right) \;dt,
 \end{multline}
where $_1F_1(\cdot)$ is the Confluent Hyper Geometric function. Based on \cite[Eq. (EH 269(5))]{edition2007table}, $a_n$ is equivalent to \eqref{eq:A_n}.
\section{PROOF OF LEMMA \ref{Lemma:p_n_rayleigh}}\label{Appendix:p_n_rayleigh}
Using \eqref{eq:BS_MGF_Rayleigh} in \eqref{eq:MGF_Cluster}, we will have:
\begin{equation}
    M_{cluster}(\eta;\; x, \;g) = \exp \left( 
(-\lambda_{BS}) \frac{-q\eta}{1-q\eta} \right)
\end{equation}
with $q= 2b_0 g \omega x^{-\alpha}$. From \cite[Eq. (8.975)]{edition2007table} we know that:
\begin{equation}
    \frac{1}{(1-\eta)^{a+1}} \exp{ \left( v \frac{-\eta }{1-\eta} \right)} = \sum\nolimits_{n=0}^{\infty} n! L_{n}^{(a)}(v) \frac{\eta^n}{n!}
\end{equation}
By setting $a=-1$, $v=-\lambda_{BS}$, and $\eta := 2b_0 g \omega x^{-\alpha} \eta$, we acquire $p_n(\lambda_{BS})$ in \eqref{eq:p_n_rayleigh}.
\section{PROOF OF LEMMA \ref{Lemma:p_n_SRF}}\label{Appendix:p_n_SRF}
Using \eqref{eq:SE_MGF_BS_SRF} in \eqref{eq:MGF_Cluster}, we will have:
\begin{equation}
    M_{cluster}(\eta;\; x, \;g) = \exp{ \left( \lambda_{BS} \sum_{j=1}^{\infty} { a_j (g\omega x^{-\alpha})^j \frac{\eta^j}{j!} } \right)}.
\end{equation}
From \cite[Eq. (1)]{orozco2021solution}, we know that:
\begin{equation}
    \exp{\left( 
          v \sum_{j=1}^{\infty} a_j \frac{\eta^{j}}{j!} 
    \right)} =1 + \sum_{n=1}^{\infty} { \frac{\eta^n}{n!} \sum_{i=1}^{n} {v^{i}B_{n,i}(a_1, ...,a_{n-i+1}) } }.
\end{equation}
By setting $v:=\lambda_{BS}$ we acquire $p_n(\lambda_{BS})$ in \eqref{eq:p_n_SRF}.

\section{PROOF OF LEMMA \ref{Lemma:MGF_RFI_ml}} \label{Appendix:MGF_RFI_ml}
Since SMAP's main-lobe antenna footprint is roughly $40^2$ km$^2$, we note that $\Lambda=40^2\lambda_c$. Thus, conditioned on $M$ we note that:
\begin{equation}
    M_{(ml)}(\eta)=\mathds{E}_{M} \left[ \mathds{E} \left[\exp\left(\eta \sum\nolimits_{\Mk{M}}T_{cluster}(\bold{x})\right) \right] \right]. \label{eq:MGF_Main_Lobe_proof}
\end{equation}
The inner $\mathds{E}[\cdot]$ is 
\Mk{the MGF of the sum of the RFI brightness temperatures of $M$ clusters, where $T_{cluster}(\bold{x})$ is defined in \eqref{eq:T_cluster}. Based on \eqref{eq:MGF_Cluster}, it can be expressed as:}
\begin{multline}
    \mathds{E} \left[\exp\left( \eta \sum\nolimits_{m}T_{cluster}(\bold{x})\right) \right]= \\\left( M_{cluster}(\eta;\; d_{(ml)}, \;g_{(ml)}) \right)^m.
\end{multline}
Accordingly, \eqref{eq:MGF_Main_Lobe_proof} can be expanded as:
\begin{equation}
    M_{(ml)}(\eta)=\sum_{m=0}^{\infty}\frac{e^{-\Lambda}\Lambda^m}{m!} \left( M_{cluster}(\eta;\; d_{(ml)}, \;g_{(ml)}) \right)^m,
\end{equation}
which is equivalent to \eqref{eq:MGF_Main_Lobe}.
\section{PROOF OF LEMMA \ref{Lemma:MGF_RFI_sl}} \label{Appendix:MGF_RFI_sl}
For ease of notation, 
we define $T_{\bold{X}_i}:=T_{cluster}(\bold{X}_i)$. Accordingly, the MGF of \eqref{eq:T_RFI_sl} is follows:
\begin{align}
    M_{(sl)}(\eta) &=\mathds{E}\left[ e^{\eta T_{(sl)}}\right] \notag \\
    & = \mathds{E}_{\Psi,\{ T_{\bold{X}_i} \}} \left[
    \prod\nolimits_{\bold{X}_i \in \Psi^{(sl)}} e^{\eta T_{\bold{X}_i}}
    \right].
\end{align}
Due to the initial assumption of the independence of clusters, we move the expectation with respect to $\{ T_{\bold{X}_i} \}$ inside the product as: 
\begin{equation}
    \mathds{E}_{\Psi_{(sl)}} \left[
         \prod\nolimits_{\bold{X}_i \in \Psi_{(sl)}}
         \mathds{E} \left[ 
           e^{\eta T_{\bold{X}_i}}
         \right]
    \right],
\end{equation}
which is the \textit{Probability Generating Functional} (PGFL) \cite{baccelli2010stochastic} of $g(\bold{x})=  \mathds{E}\left[ e^{\eta T_{\bold{x}}} \right]$  over the set $\mathcal{B}_{(sl)}$ and can be written as:
\begin{align}
    \mathcal{P}_{\Psi_{(sl)}}(g) &= \mathds{E}_{\Psi_{(sl)}} \left[ \prod\nolimits_{\bold{X}_{i} \in \Psi_{(sl)}} g(\bold{X}_{i}) \right]
    \notag \\ 
    &= \exp \left( 
           -\int_{\mathcal{B}_{(sl)}} \left( 1-g(\bold{x}) \right) \, \Lambda(d\bold{x})
    \right), \label{eq:PGFL}                     
\end{align}
where, based on Figure \ref{fig:smap_exposed}, in spherical coordinates for equidistant points to SMAP for a PPP with intensity $\lambda_c$:
\begin{equation}
    \Lambda \left( d \bold{x} \right) = 2 \pi r^2 \lambda_c \sin(\theta) \; d\theta.   \label{eq:Lambda_dX}
\end{equation}
We note that $g(\bold{x})$ is the MGF of RFI brightness temperature of a cluster at point $\bold{x} \in \mathcal{B}_{(sl)}$ as in \eqref{eq:MGF_Cluster}, which we denote here with $g(\bold{x})=M_{cluster}(\eta; \; \lVert \bold{x}-\bold{h} \lVert, \;g_{(sl)} )$, where $x=\lVert \bold{x}-\bold{h} \lVert$ is the distance to SMAP. From Figure \ref{fig:smap_exposed}, and using the law of cosines, we note that:
\begin{equation}
    x= \left( r_e^2 + h^2 -2hr_e \cos (\theta) \right)^{\frac{1}{2}}, 
\end{equation}
with:
\begin{align}
    &dx = hr_e \sin(\theta) \left( r_e^2 + h^2 -2hr_e \cos (\theta) \right)^{-\frac{1}{2}}\; d\theta \notag \\
    \Rightarrow x\;&dx= hr_e \sin(\theta)\; d\mathtt{\theta}. \label{eq:xdx}
\end{align}
Comparing \eqref{eq:Lambda_dX} and \eqref{eq:xdx}, we note that:
\begin{equation}
    \Lambda \left( d \bold{x} \right) = 2 \pi \left( \frac{r_e}{h} \right) \lambda_cx \;dx, \label{eq:lambda_to_x}
\end{equation}
where $x$ is in the range of $d_{min}$ and $d_{max}$. By substituting \eqref{eq:lambda_to_x} in \eqref{eq:PGFL}, we will have \eqref{eq:MGF_Side_Lobe}.
\section{PROOF OF LEMMA \ref{Lemma:Laplace_P}} \label{Appendix:Laplace_P}

\Mk{ $\mathcal{L}_P(s)$ in (\ref{eq:s_e}) can be expanded as: }

\begin{equation}
    \mathds{E}_Z \left[ \mathcal{L}_P(s|Z) \right] = \int_{0}^{\infty}f_{Z}(z) \mathcal{L}_P(s|z) dz,
\end{equation}
where:
\begin{align}
    \mathcal{L}_P(s|z) &= \mathds{E}[e^{\Mk{-}sI}] \notag \\
    & = \mathds{E}_{ \{R_i\},\{H_i\} } \left[ 
    \prod\nolimits_{\{i\}} \exp(-sp_{tx}R_i^{-\alpha}H_{i})
    \right], \notag 
\end{align}
which, if the \Mk{number of BSs} $N$ is \Mk{generated} according to a Poisson distribution with \Mk{intensity $\lambda_{BS}$} \Mk{and the distance of the BS from a UE is distributed according to $f_{R}(r|z)$}, then according to the law of PGFL is equal to:
\begin{equation}
    \exp{\left( -\lambda_{BS} \int_{0}^{\infty} 
    (1- \mathds{E}_{H}[e^{(-sp_{tx}r^{-\alpha}H)}]) f_{R}(r|z) \;dr
    \right) }
\end{equation}
and $\mathds{E}_{H}[\cdot]$ is the Laplace function of an exponential distribution. Finally, we will have \eqref{eq:Laplace_P}.

\ifCLASSOPTIONcaptionsoff
  \newpage
\fi

\bibliographystyle{IEEEtran}
\bibliography{Main}

\end{document}